\documentclass[a4paper, 11pt]{article}

\usepackage[english]{babel}

\usepackage{geometry} 
\newgeometry{
    top=1in,
    bottom=1in,
    outer=1in,
    inner=1in,
}

\usepackage{setspace}
\usepackage{color}

\usepackage{amsmath,amsthm}
\usepackage{amssymb}
\usepackage{graphicx}
\usepackage{subcaption}
\usepackage{braket}
\usepackage{mathtools}
\usepackage{enumitem}
\usepackage{siunitx}
\usepackage{authblk}
\usepackage{cite}

\usepackage[colorlinks=true, allcolors=blue]{hyperref}

\newcommand{\dket}[1]{\ket{#1} \rangle}
\newcommand{\dbra}[1]{\langle \bra{#1}}

\makeatletter
\newcommand{\doublewidetilde}[1]{{%
  \mathpalette\double@widetilde{#1}%
}}
\newcommand{\double@widetilde}[2]{%
  \sbox\z@{\(\m@th#1\widetilde{#2}\)}%
  \ht\z@=.9\ht\z@
  \widetilde{\box\z@}%
}

\newcommand{\map}[1]{\mathcal{#1}}
\newcommand{\smap}[1]{\widetilde{\mathcal{#1}}}

\newcommand{\lnsp}[1]{{\mathcal{#1}}}
\newcommand{\opsp}[1]{{\mathcal{L}(#1)}}
\DeclareMathOperator{\Tr}{Tr}

\newcommand{\csp}{\mathbb{C}}
\newcommand{\maptp}{\mathcal{A}} 

\newcommand{\memory}[1]{{\mathcal{M}_{#1}}}
\newcommand{\hil}[1]{{\mathcal{H}_{#1}}}
\newcommand{\lin}{\mathcal{L}}

\newcommand{\iden}{\mathrm{id}}
\newcommand{\ketbra}[1]{{\ket{#1} \bra{#1}}}
\newcommand{\tr}{\mathrm{Tr}}
\newcommand{\id}{\mathbb{I}}

\makeatother

\theoremstyle{definition}
\newtheorem{theo}{Theorem}
\newtheorem{defi}{Definition}
\newtheorem{lemm}{Lemma}

\newtheorem{rema}{Remark}

\title{Probabilistic Storage and Retrieval of Quantum Superchannels for ``Retrospective'' Intervention}

\author[1]{Wataru Yokojima}
\author[1,2]{Jisho Miyazaki}
\author[1,3]{Mio Murao}

\affil[1]{Department of Physics, Graduate School of Science, The University of Tokyo}
\affil[2]{Ritsumeikan University BKC Research Organization of Social Sciences}
\affil[3]{Trans-scale Quantum Science Institute, The University of Tokyo}

\date{\today}

\begin{document}

\maketitle

\begin{abstract}
    Storing an unknown quantum computation in a quantum state and retrieving it at a desired later time is a challenging task, hindered by the no-programming theorem of quantum computations. In the previous studies on the task of \textit{probabilistic storage-and-retrieval (pSAR) of quantum channels}, the maximum probability of exactly retrieving a single unknown unitary \textit{channel} from a quantum state in which the unknown unitary has been encoded via multiple calls to the unknown unitary channel is derived. In this work, we consider a higher-order version of pSAR, the probabilistic storage-and-retrieval of definite-causal unitary \textit{superchannels}, which are physically modeled by sequences of unitary channels with open slots where arbitrary channels can be inserted between the unitary channels for intervention. This task requires activating the ``retrospective'' intervention functionality on the superchannel, beyond its normal intervention functionality.
    We propose two protocols: \textit{partial teleportation}, which is supported by our numerical results for the small queries, and \textit{staircase backstitch}, which achieves unit success probability asymptotically as the number of queries increases. We also derive a universal inversion protocol for unitary superchannels.
\end{abstract}


\section{Introduction}
A convenient feature of programmable processors is their ability to store programs and execute them at any desired time later.
The quantum extensions of programmable processors were studied in many different scenarios \cite{NPnielsenProgrammableQuantumGate1997,NPvidalStoringQuantumDynamics2002,NPzimanRealizationPositiveoperatorvaluedMeasures2005,NPbergouUniversalProgrammableQuantum2005,NPdarianoEfficientUniversalProgrammable2005,NPhilleryApproximateProgrammableQuantum2006,NPperez-garciaOptimalityProgrammableQuantum2006,ishizaka2008asymptotic,NPkubickiResourceQuantificationNoPrograming2019,NPyangOptimalUniversalProgramming2020,NPgschwendtnerProgrammabilityCovariantQuantum2021,NPgschwendtnerInfiniteDimensionalProgrammableQuantum2021}, which raised no-programming theorems, highlighting the limits of their performance.
Nevertheless, at the cost of success probability or accuracy, it is possible to write a quantum program whose length is set by the desired precision and execute the computation later at the desired time.

Whether in classical computation or quantum computation, the programmer needs to be aware of the concise description of the computation.
It is not necessarily possible to write down the entire program at once.  The agent may need to suspend computation when the computation is provided as a black box and is unknown to the agent. In quantum computation, the inherent indistinguishability of operations poses a particularly challenging problem for dealing with unknown quantum channels. Bisio \textit{et al.} \cite{bisio2010optimallearning} studied the approximate accuracy of deterministically retrieving a unitary channel from quantum states prepared by multiple calls to black-box unitary channels.

Framing the storage-and-retrieval of unknown computation in a probabilistic exact setting reveals a clear distinction from deterministic approximate settings. Probabilistic storage-and-retrieval (pSAR) protocols were recognized early on as a method to circumvent the no-programming theorem by introducing a relaxation to probabilistic success
\cite{NPnielsenProgrammableQuantumGate1997,NPvidalStoringQuantumDynamics2002}. Sedlak \textit{et al.} \cite{sedlak2019PSAR} studied the success probability in storing the action of an unknown unitary channel from multiple calls to the black box implementing the unknown unitary channel and in retrieving the exact action later.
The maximum success probability is achieved by the quantum teleportation protocol \cite{bennetetal1993teleportation,gottesmanDemonstratingViabilityUniversal1999a} when only a single use of the black box is allowed and by the port-based teleportation protocol \cite{ishizaka2008asymptotic,ishizaka2009quantum,studzinski2017port,mozrzymas2018optimal} when multiple calls to the black boxes are available. This pSAR task is then formulated for unknown computations that contain restricted classes of unitary gates \cite{sedlak2020PSARphasegate,sedlak2024twounitaries}.

In this work, we initiate the study of pSAR tasks for unknown quantum computations with \textit{intervention functionality}. This is equivalent to considering pSAR tasks for \textit{quantum superchannels}, also known as quantum networks and quantum combs \cite{chiribella2008architechture,chiribella2008transforming,chiribella2009framework}.
Quantum superchannels are a mathematical model of higher-order quantum computation, describing transformations that map quantum channels to quantum channels. Superchannels with definite causal order are physically implemented as concatenations of quantum memory channels, allowing the computing agent to intervene between successive stages \cite{chiribella2008architechture,chiribella2008transforming,chiribella2009framework}.
Unlike a mere sequence of channels applied at intervals, a quantum superchannel includes memory subsystems that are inaccessible to the agent.

What appears challenging about the pSAR of unitary superchannels is that the temporal structure of the computation must be stored in and retrieved from quantum states, which are inherently static. That is, starting the retrieval at the desired time is a part of the
task; we need the ability to pause-and-resume in the retrieval part so that other operations can intervene while paused. Therefore, the result of \cite{sedlak2019PSAR} for the channel pSAR does not simply generalize to the superchannel pSAR.

SAR protocols of quantum superchannels effectively allow the input operations on the superchannel to intervene ``retrospectively'' in the computation.
SAR protocols for quantum channels and superchannels both aim to delay the application of the unknown input state of the computation. The (probabilistic) port-based teleportation \cite{ishizaka2008asymptotic,ishizaka2009quantum,studzinski2017port,mozrzymas2018optimal} is an optimal protocol to delay the input states for channel pSAR \cite{sedlak2019PSAR}. In this context, protocols for superchannel pSAR must delay not only unknown input states but also unknown intervening operations inserted in the superchannels. Our study will demonstrate the difference between states and operations in terms of the difficulty of delaying them, in other words, retrospectively applying them. As a byproduct of our analysis of pSAR for unitary superchannels, we also present a protocol for universally inverting a unitary superchannel while preserving its intervention structure.

While quantum superchannels, in their broadest sense, encompass quantum computation over indefinite causal structures, this paper focuses on pSAR of superchannels with a definite causal structure. Among the definite-causal quantum superchannels, we particularly focus on unitary ones, that is, those constituting unitary channels with memories.
We first analytically show that the pSAR of unitary superchannels is possible with non-zero probability without error by constructing two pSAR protocols.
Maximum success probabilities are evaluated by numerical calculations for small-sized problems.

We begin in Section~\ref{sec:formalizing_pSAR} by formalizing pSAR tasks for unitary superchannels and related channels, which we refer to as \textit{unitary staircases}. In Section~\ref{sec:protocols}, we then propose two protocols for achieving the pSAR of unitary superchannels. We also present a construction of a protocol for inverting unitary superchannels. In Section~\ref{sec:probability}, we report the results of our numerical calculation of the maximum success probability of pSAR, together with a brief summary of the calculation method. Finally, we conclude in Section~\ref{sec:conclusion} and discuss several open questions.

\section{Formalizing pSAR}\label{sec:formalizing_pSAR}
\subsection{Quantum superchannels and quantum staircases}
A quantum superchannel here refers to a sequence of quantum instruments with some of their subsystems interconnected by memory channels.
Although such structures are occasionally referred to as having \textit{definite causal order}, we refer to them simply as quantum superchannels (or superchannels) in this work.

More precisely, let $\hil{k}$ be Hilbert spaces for $k=0,\ldots ,2K+1$.
A $K$-slot quantum superchannel $\smap{C}$ of type $(\hil{0},\ldots,\hil{2K+1})$ is an operation represented by a concatenation of quantum instruments $\map{C}_k :\lin(\memory{k-1} \otimes \hil{2k}) \rightarrow \lin(\memory{k} \otimes \hil{2k+1})$ for $k=0,\ldots,K$.
In this configuration, the memory spaces $\memory{k}$ between $\map{C}_k$ and $\map{C}_{k+1}$ are linked by identity operations for $k=0,\ldots,K-1$, while $\memory{-1}$ and $\memory{K}$ are defined as $\mathbb{C}$ (See Figure \ref{fig:superchannel}).
The quantum superchannel permits some interventions after the execution of $\map{C}_k$ and before the execution of $\map{C}_{k+1}$, in the free subsystem from $\hil{2k+1}$ to $\hil{2k+2}$.
These intervals between free subsystems in which the superchannel receives interventions are referred to as (open) slots of the superchannel.
The type of a quantum superchannel is often represented by the dimensions of its associated Hilbert spaces, $(\dim \hil{0},\ldots,\dim \hil{2K+1})$.
\begin{figure}[htbp]
  \centering
  \includegraphics[width=0.3\columnwidth]{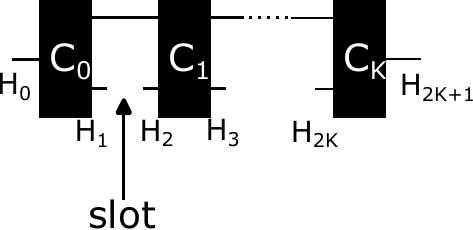}
  \caption{A $K$-slot quantum superchannel of type $(\dim \hil{0},\ldots,\dim \hil{2K+1})$ containing channels $\map{C}_1$ to $\map{C}_K$.}
  \label{fig:superchannel}
\end{figure}

The free subsystems included in the domains and codomains of any map are termed the input and output ports of the superchannel.
That is, $\hil{0},\hil{2},\ldots \hil{2K}$ are input ports and $\hil{1},\hil{3},\ldots \hil{2K+1}$ are output ports in the current example.

The target of storage-and-retrieval in this article is unitary superchannels. A unitary superchannel is a quantum superchannel that only contains unitary channels $\map{U}_k:\lin(\memory{k-1} \otimes \hil{2k}) \rightarrow \lin(\memory{k} \otimes \hil{2k+1})$. The set of all $K$-slot unitary superchannels of type $(d_1,~\ldots,d_{2K+1})$ is denoted by $\smap{U}[d_1,~\ldots,d_{2K+1}]$.

To assess the difficulty of activating the retrospective intervention functionality in the pSAR setting, we compare quantum superchannels and their reductions without the intervention functionality. We introduce \textit{quantum staircases} to represent the latter. If $\smap{C}$ is a quantum superchannel of type $(\dim \hil{0},\ldots,\dim \hil{2K+1})$, define the quantum staircase $\map{C}$ as a quantum instrument from the composition of input ports $\hil{0} \otimes \hil{2} \otimes \cdots \otimes \hil{2K}$ to that of output ports $\hil{1} \otimes \hil{3} \otimes \cdots \otimes \hil{2K+1}$, which is simply obtained by first preparing ancillary systems for input ports and then intervening with swap operations at each slot, exchanging the output port to the ancilla of input port as in Figure~\ref{fig:staircase}.\footnote{This swap exchanges systems of different sizes in general: $d_1 \times d_2 \rightarrow d_2 \times d_1$.} Since a quantum staircase is a channel, its input should be prepared simultaneously as a multipartite state, and there is no room to intervene with temporally ordered operations $\map{U}_0$ to $\map{U}_K$.
\begin{figure}[htbp]
  \centering
  \includegraphics[width=0.5\columnwidth]{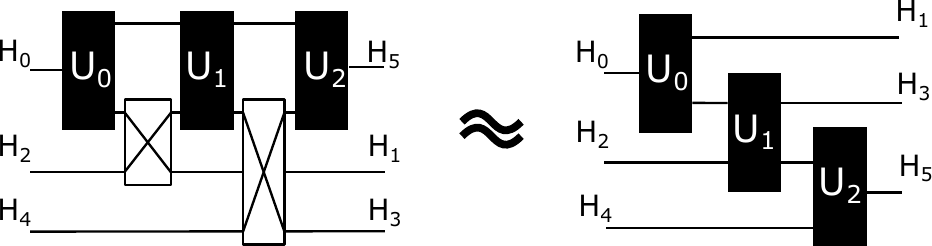}
  \caption{Introducing ancillary systems and intervening with swap operations to obtain a staircase from the ($2$-slot) superchannel.}
  \label{fig:staircase}
\end{figure}

We remove the tilde $\tilde{}$ on top of superchannels $\smap{C}, ~ \smap{D}, ~ \smap{E}$ to represent the corresponding staircases $\map{C},~\map{D},~\map{E}$. A unitary staircase refers simply to a quantum staircase constructed from a unitary superchannel. The set of all unitary staircases made from superchannels of type $(d_1,~\ldots,d_{2K+1})$ is denoted by $\map{U}[d_1,~\ldots,d_{2K+1}]$.

\subsection{Storage-and-retrieval configurations}
The pSAR task of unitary superchannels is formulated analogously to that of unitary channels. The storage and retrieval circuits are treated as quantum superchannels as shown in Figure \ref{fig:pSARsuperchannel}. In this work, we do not allow an indefinite causal structure in the storage and retrieval circuits.
A storage superchannel can accommodate $N$ identical unknown unitary superchannels $\smap{U}$ as sequences of black boxes, and produce a quantum state $\sigma_{\smap{U}}$ in the memory system $\memory{}$. A retrieval superchannel receives the state $\sigma_{\smap{U}}$ and exactly implements the original unitary superchannel $\smap{U}$ with some probability. We say a pair of storage and retrieval superchannels achieves the pSAR of unitary superchannels of type $(d_0,\ldots,d_{2K+1})$ with probability $p$ when the retrieval is successful for any unitary superchannel of type $(d_0,\ldots,d_{2K+1})$ and with the same overall probability $p$.
\begin{figure}[htbp]
  \centering
  \includegraphics[width=0.9\columnwidth]{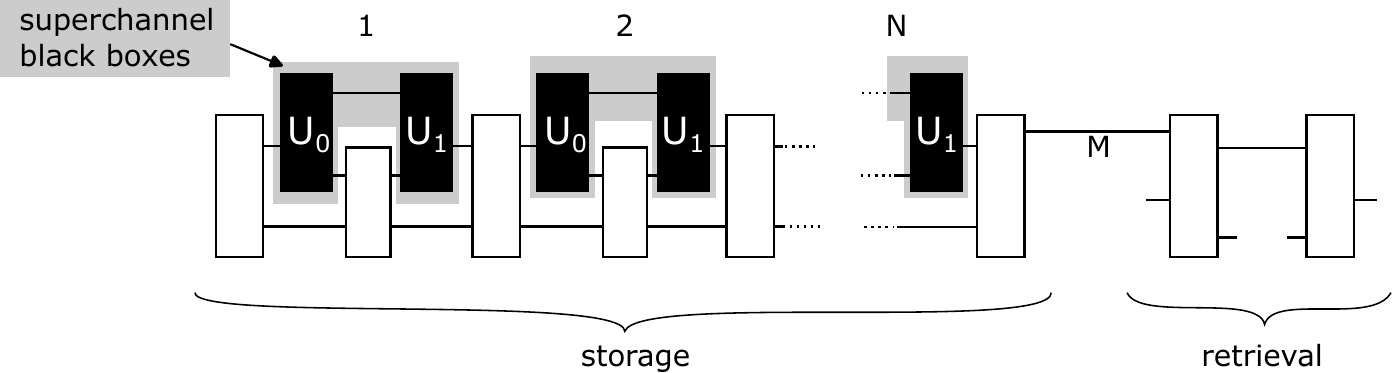}
  \caption{The storage and retrieval superchannels for pSAR of $1$-slot unitary superchannels.}
  \label{fig:pSARsuperchannel}
\end{figure}

The performance of pSAR is sensitive to the specific configuration of the storage and retrieval circuit, with several distinct paradigms under consideration. The distinction between unitary superchannels and unitary staircases as targets for pSAR is a key element of this work.

The pSAR task for unitary staircases can be formulated as a special case of the pSAR for unitary channels \cite{sedlak2019PSAR}. Specifically, the objective is to reconstruct a single unknown unitary staircase from a quantum state prepared by a storage circuit that queries the unknown staircase $N$ times. Since unitary staircases constitute a subclass of unitary channels, existing pSAR protocols for unitary channels \cite{sedlak2019PSAR} are directly applicable, though they may not be optimal for this specific subclass. Prior studies \cite{sedlak2020PSARphasegate, sedlak2024twounitaries} have demonstrated that the maximum success probability can be enhanced by restricting the pSAR target to specific subclasses of unitaries. In Section \ref{sec:probability}, we numerically evaluate the maximum success probability for the pSAR of unitary staircases.

In the pSAR protocol for unitary superchannels, the objective is to reconstruct the unknown superchannel, including its internal slots for intervention. The storage process queries a black box that implements this unknown unitary superchannel. Unlike a standard channel, a superchannel black box is characterized by its extended operation, which spans multiple computational steps until all interventions are completed. Specifically, the storage circuit treats each black-box query as an instance of a superchannel of the prescribed type. As illustrated in Figure~\ref{fig:pSARsuperchannel}, we assume that the storage circuit is provided with the full superchannel for each call.

In addition to the two configurations discussed above, we introduce the superchannel-to-staircase pSAR—a protocol designed to retrieve unitary staircases using unitary superchannels during the storage phase. The objective is to reconstruct the unitary staircase corresponding to an unknown superchannel. The storage process queries the superchannel black box following the same constraints as superchannel pSAR. Since the superchannel-to-staircase pSAR captures all features of the superchannel except for its intervention functionality, comparing it with superchannel pSAR directly isolates the inherent difficulty of enabling retrospective intervention.

A further distinction in SAR protocols concerns whether the $N$ black-box queries are performed in parallel or sequentially over time.\footnote{In a more general setting, the storage process could use the black boxes in an indefinite causal order. In this work, however, we focus on SAR circuits with a definite causal structure and leave such generalizations for future work.}
In this article, we focus on storage of the latter kind, referred to as \textit{sequential storage}, which is more general than parallel storage.
As illustrated in Figure~\ref{fig:pSARsuperchannel}, the sequential storage considered here invokes each black-box superchannel only after all interventions on the preceding one have been completed.\href{}{}

The superchannel pSAR is rigorously formulated as follows. Let the type of unitary superchannels be $(d_0,\ldots,d_{2K+1})$. The storage that can contain $N$ unknown unitary superchannels as sequences of black boxes is itself a superchannel of the type
\begin{equation}
    \label{eq:type_storage} (1, ~[d_0,\ldots,d_{2K+1}]_N,~\dim \memory{}),
\end{equation}
where $[d_0,\ldots,d_{2K+1}]_N$ is the shorthand for the $N$-repetition of the sequence $d_0,\ldots,d_{2K+1}$. The storage superchannel passes a unitary-dependent state $\sigma_{\smap{U}}$ on system $\memory{}$ to the retrieval part.
The retrieval part is also a superchannel of the type
\begin{equation}
    \label{eq:type_retrieval}   (d_0 \times \dim \memory{},~d_1,\ldots,d_{2K+1}).
\end{equation}
Once it receives the state $\sigma_{\smap{U}}$ on system $\memory{}$ from the storage, the retrieval superchannel should probabilistically turn into the original unitary superchannel $\smap{U}$.

The success probability and the realization of pSAR are formally defined as follows.
\begin{defi}\label{def:pSAR}
    A pair of storage superchannel $\smap{S}$ of type \eqref{eq:type_storage} and retrieval superchannel $\smap{R}$ of type \eqref{eq:type_retrieval} is defined to realize pSAR of unitary superchannels of type $(d_0,\ldots,d_{2K+1})$ with probability $p$ when
    \begin{equation}
        \label{eq:pSAR} \smap{R}(\sigma_{\smap{U}}) = p ~ \smap{U} \qquad \left( \sigma_{\smap{U}} = \smap{S}(\smap{U}, \ldots, \smap{U} ) \right),
    \end{equation}
    holds for any unitary superchannel $\smap{U} \in \smap{U}[d_0,\ldots,d_{2K+1}]$.
\end{defi}
Note that $\smap{S}(\smap{U}, \ldots, \smap{U} )$ represents the state obtained by inserting $N$ unitary superchannels $\smap{U}$ into $\smap{S}$, and $\smap{R}(\sigma_{\smap{U}})$ represents the quantum superchannel obtained by passing state $\sigma_{\smap{U}}$ at $\memory{}$ of $\smap{R}$. It is implicit in Eq.~\eqref{eq:pSAR} that the probability $p$ does not depend on $\smap{U}$ and on the other inputs of the quantum superchannel $\smap{R}(\sigma_{\smap{U}})$.

We define the maximum success probability of superchannel pSAR as
\begin{equation}
    p_{\mathrm{max},N} \left(d_0,\ldots,d_{2K+1} \right) := \max_{(\smap{S},\smap{R})} p,
\end{equation}
where the pair $(\smap{S},\smap{R})$ realizes the pSAR of unitary superchannels of type $(d_0,\ldots,d_{2K+1})$. We evaluate the maximum success probabilities by numerical calculations for problems of small sizes.

The maximum success probabilities of staircase and superchannel-to-staircase pSARs can be defined analogously to that of superchannel pSAR. Let us denote them as $p_{\mathrm{max},N}^{\map{U} \rightarrow \map{U}}$, $p_{\mathrm{max},N}^{\smap{U} \rightarrow \map{U}}$ and $p_{\mathrm{max},N}^{\smap{U} \rightarrow \smap{U}}$, respectively. Since a unitary staircase can be physically derived from its corresponding superchannel at no cost by the procedure of Figure~\ref{fig:staircase}, the following inequalities hold:
\begin{align}
    \label{eq:inequality1}  & p_{\mathrm{max},N}^{\map{U} \rightarrow \map{U}} (d_0,\ldots,d_{2K+1}) \leq p_{\mathrm{max},N}^{\smap{U} \rightarrow \map{U}} (d_0,\ldots,d_{2K+1}), \\
    \label{eq:inequality2}  & p_{\mathrm{max},N}^{\smap{U} \rightarrow \smap{U}} (d_0,\ldots,d_{2K+1}) \leq p_{\mathrm{max},N}^{\smap{U} \rightarrow \map{U}} (d_0,\ldots,d_{2K+1}).
\end{align}
Consequently, the difficulty inherent in retrospective intervention is quantitatively characterized by the performance gap between $p_{\mathrm{max},N}^{\smap{U} \rightarrow \map{U}}$ and $p_{\mathrm{max},N}^{\smap{U} \rightarrow \smap{U}}$.

\section{Protocols for superchannel pSAR}\label{sec:protocols}
In this section, we propose two protocols for pSAR of unitary superchannels. Both protocols are based on the idea of incorporating the pSAR of unitary channels as a subroutine. We first describe this idea in detail in Section \ref{subsec:climb} and then proceed to the protocols.

\subsection{Climbing up and down the hierarchy in higher-order quantum computation}\label{subsec:climb}
The protocols proposed in this work decompose the pSAR of unitary superchannels into three distinct steps:
\begin{itemize}
    \item The unknown superchannels are transformed into their corresponding unitary staircases. This procedure requires only knowledge of the superchannel type and is implemented by preparing ancillary systems and intervening at the slots of the black boxes with swap operations (see Figure~\ref{fig:staircase}).
    \item The pSAR protocol for unitary channels is invoked multiple times to store and retrieve a sufficient number of unitary staircases. This retrieval process is probabilistic and requires multiple queries to the black box channels.
    \item The retrieved unitary staircases are converted back into a single superchannel. This final step restores the intervention functionality.
\end{itemize}
Among existing channel pSAR protocols, we provide a brief review of the optimal scheme based on port-based teleportation (PBT) in Appendix \ref{sec:PBT}.
Detailed protocols for the third step are presented in the following subsections.

This idea is rooted in the hierarchy of higher-order quantum operations, as illustrated in Figure~\ref{fig:higher-order}. In this hierarchy, quantum states constitute the fundamental $0$th-order objects. Quantum channels, which map states to states, are classified as first-order objects. Quantum superchannels represent even higher-order objects, as they act upon (super)channels to transform them into other (super)channels.
\begin{figure}
    \centering
    \includegraphics[width=0.4\linewidth]{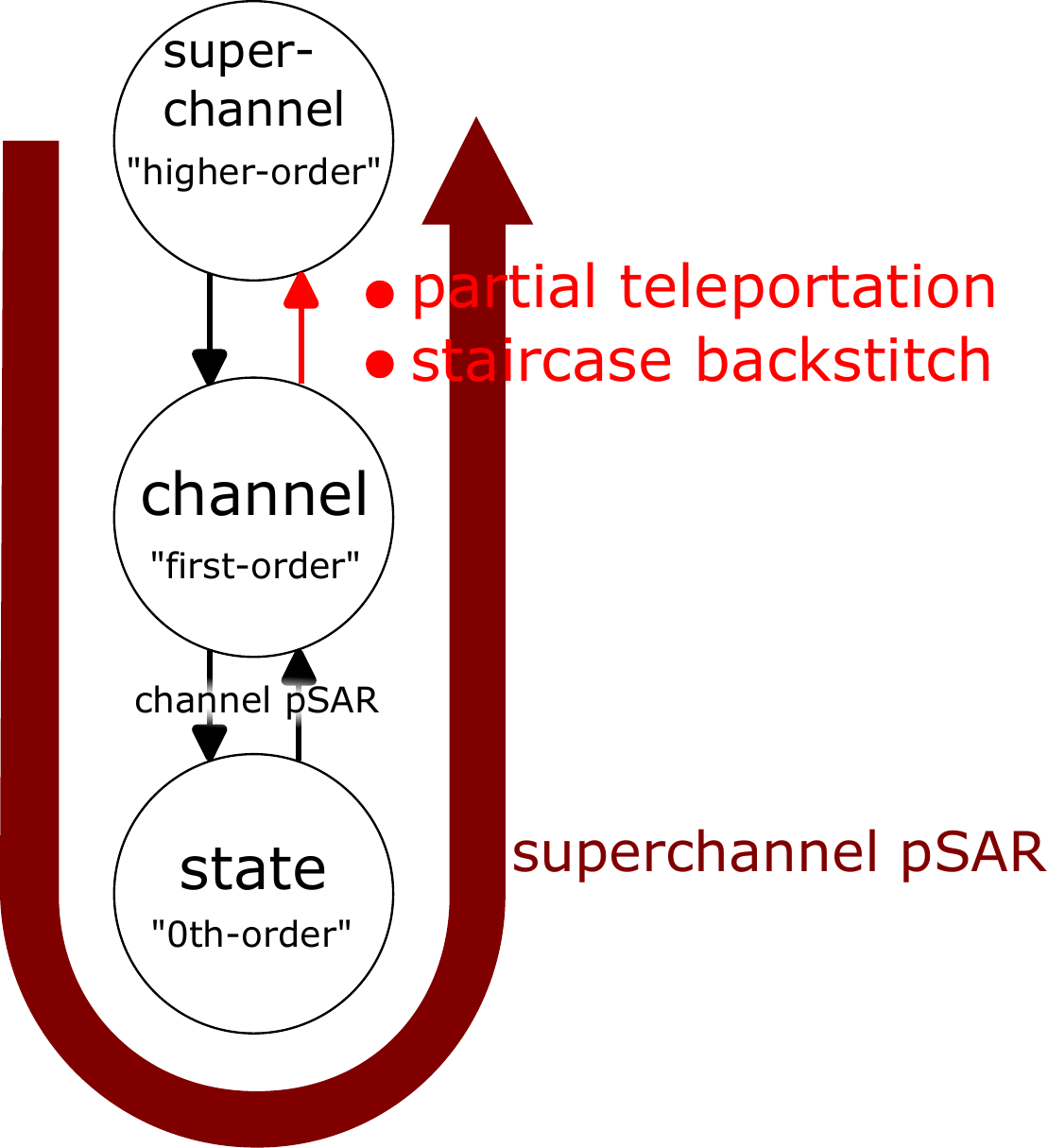}
    \caption{Hierarchy in higher-order computation and how the pSAR protocols climb it up and down.}
    \label{fig:higher-order}
\end{figure}

SAR can be understood as a procedure for encoding higher-order objects (superchannels) into $0$th-order objects (states) and subsequently recovering them. The three-step method described above realizes SAR by systematically descending and ascending the hierarchy, rather than attempting a direct leap between the higher-order and $0$th-order layers.

In our current framework, we convert a unitary superchannel $\smap{U}$ first to a staircase $\map{U}$ (a first-order object), and then to its corresponding storage state $\sigma_{\map{U}}$ ($0$th-order). Because the physical implementation of the conversion $\smap{U} \rightarrow \map{U}$ and the channel pSAR $\map{U} \rightarrow \sigma_{\map{U}} \rightarrow \map{U}$ are already established, the remaining challenge is the inverse transformation from staircases back to superchannels: $\map{U} \rightarrow \smap{U}$. This requires a mechanism to reactivate the intervention functionality, thereby lifting the first-order object back into the higher-order domain.

\subsection{Protocol 1: partial teleportation}\label{subsec:partial_teleportation}
The first protocol uses post-selected quantum teleportation to delay input, as depicted in Figure \ref{fig:partial_tele}. To obtain $\smap{U}$ from $\map{U}$, first apply $\map{U}$ to one side of the maximally entangled states $\otimes_{k \geq 2,\mathrm{even}} \Psi_{\hil{k}}$. At any time obtaining after the output subsystem $\hil{k}$ with odd $k$, the agent can interrupt with other operations and then perform the maximally-entangled measurement over the bipartite system $\hil{k+1} \otimes \hil{k+1}$, half of which is from the output of the interruption, and the other half is from the maximally entangled state $\Psi_{\hil{k+1}}$. The teleportation succeeds with probability $1/(\dim \hil{k+1})^2$. By repeating this procedure for even $k$s in the increasing order from $k=2$, the agent succeeds in the transformation $\map{U} \mapsto \smap{U}$ with the joint probability
\begin{equation}
    \Pi_{k=2,\mathrm{even}}^{2K} \frac{1}{\dim \hil{k}^2}.
\end{equation}
\begin{figure}
    \centering
    \includegraphics[width=0.5\linewidth]{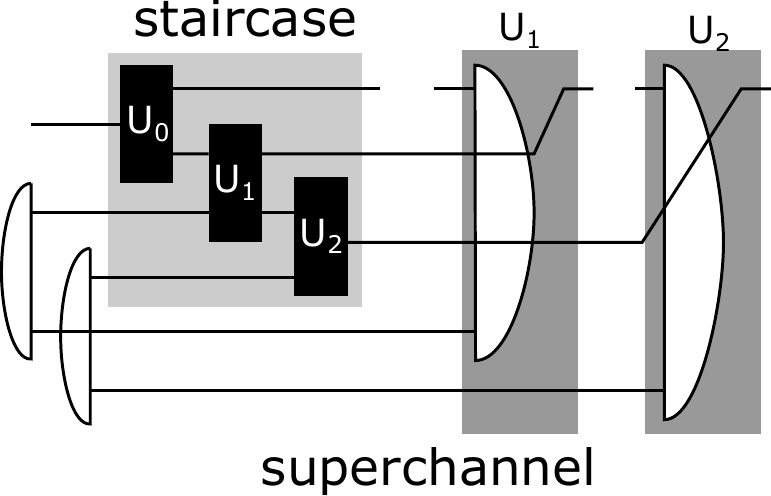}
    \caption{The channel-to-superchannel conversion protocol based on partial teleportation. The prepared states are all maximally entangled states $\ket{\Psi} := \sum_{i=1}^d \ket{ii}/ \sqrt{d}$, and all the measurements include the element $\ket{\Psi} \bra{\Psi}$ corresponding to the success branch.}
    \label{fig:partial_tele}
\end{figure}

We combine the above partial teleportation with the optimal channel pSAR protocol \cite{sedlak2019PSAR}, such as PBT (see Appendix \ref{sec:PBT} for a brief review). The combined protocol for pSAR of unitary superchannels is presented by the quantum circuit of Figure \ref{fig:pSAR_tele}. The success probability is given by
\begin{equation}
    \label{eq:p_teleportation}  p^\text{tele}_N = \frac{N}{N-1 + \dim \hil{0}^2} \times \Pi_{k=2,\mathrm{even}}^{2K} \frac{1}{\dim \hil{k}^2},
\end{equation}
where the factor $N / (N-1 + \dim \hil{0}^2)$ originates from the channel pSAR.
\begin{figure}[htb]
    \centering
    \includegraphics[width=0.5\textwidth]{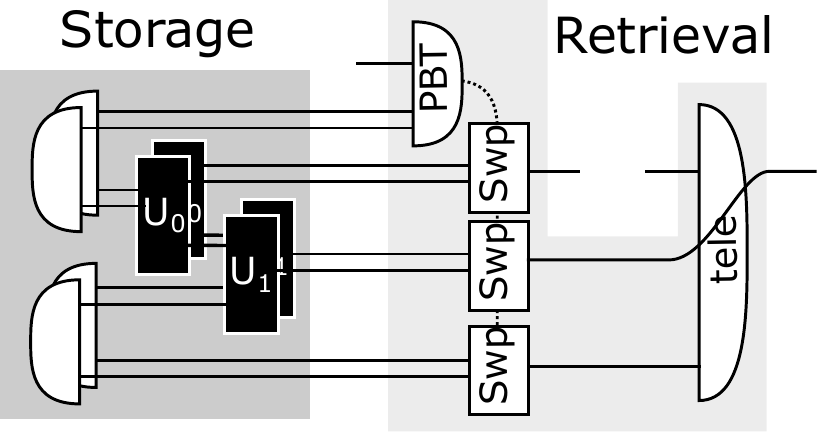}
    \caption{The pSAR protocol of unitary superchannels based on the partial teleportation and the port-based teleportation (PBT). We depict the simplest but nontrivial case $N=2$ and $K=1$. The storage circuit prepares maximally entangled states. The retrieval phase executes the measurement of PBT at the first tooth of the superchannel, and that of usual teleportation at the second tooth.}
    \label{fig:pSAR_tele}
\end{figure}

This protocol has several advantages as well as limitations.
First, only a single staircase $\map{U}$ is consumed in the conversion from $\map{U}$ to $\smap{U}$.
Second, if we employ PBT for channel pSAR, the overall protocol can be extended to pSAR for general quantum superchannels that are not necessarily unitary. This is because both PBT and partial teleportation are applicable to general channels and superchannels. On the other hand, owing to the use of partial teleportation, the success probability $p^\text{tele}N$ incurs a constant overhead $\Pi{k=2,\mathrm{even}}^{2K} \frac{1}{\dim \hil{k}^2}$, regardless of how many black boxes of the unknown superchannel are available.  In what follows, we present our second pSAR protocol, which outperforms the partial-teleportation-based protocol for large $N$.

\subsection{Protocol 2: staircase backstitch}

The second protocol utilizes a higher-order quantum transformation \cite{tarantoHigherOrderQuantumOperations2025} called \textit{unitary inversion} \cite{quintino2019reversing,yoshida2023reversing,chenQuantumAlgorithmReversing2025}. Given enough calls to the black box implementing a unitary channel $\map{U}$, we show that it is possible to construct a \textit{deterministic} circuit that implements the inverse channel $\map{U}^{-1}$ \cite{yoshida2023reversing,chenQuantumAlgorithmReversing2025}.

To simplify the discussion, we first examine the case of $1$-slot unitary superchannels.
If $\map{U}$ is a unitary staircase, the concatenation of $\map{U}$, $\map{U}^{-1}$, and again $\map{U}$ presented in Figure \ref{fig:backstitch} (a) contains the superchannel $\smap{U}$ as in Figure \ref{fig:backstitch} (b).
Therefore, it is possible to implement the staircase-to-superchannel conversion by the circuit of Figure~\ref{fig:backstitch} (a), where the inversion at the middle is implemented by the superchannel for unitary inversion, using multiple calls to the unitary staircase. 
We call this strategy to implement the channel-to-superchannel conversion \textit{staircase backstitch}.
\begin{figure}[htb]
    \centering
    \includegraphics[width=0.4\textwidth]{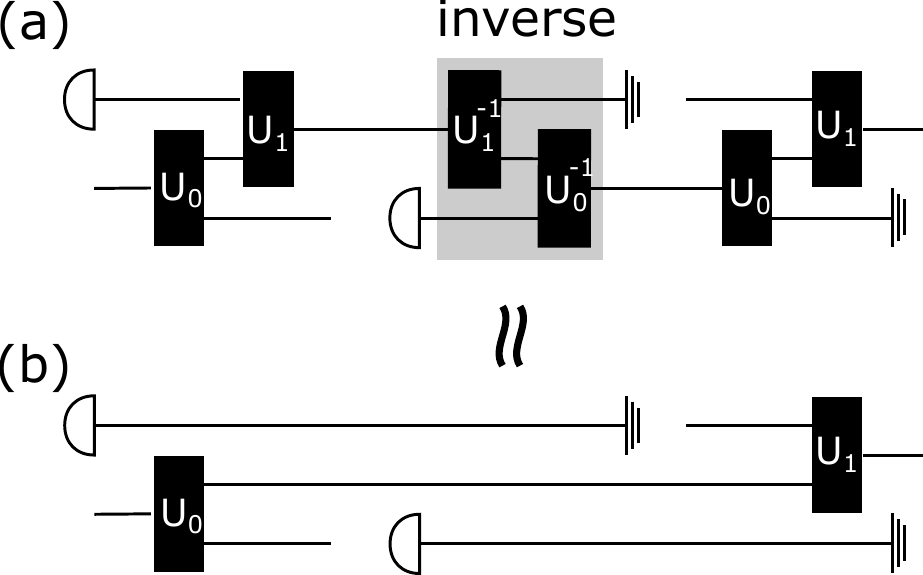}
    \caption{The channel-to-superchannel transformation for $1$-slot superchannels based on staircase backstitch. The circuit (a), consisting of a sequence of channels, is equivalent to circuit (b), which includes the corresponding superchannel. The auxiliary systems can be prepared in any state since they are isolated from the superchannel and are eventually discarded at the grounding symbols.}
    \label{fig:backstitch}
\end{figure}

We note that, in the channel-to-superchannel conversion used in the staircase backstitch protocol, the input to $\map{U}_1$ of the unitary staircase $\map{U}$ is fixed. Consequently, the protocol effectively consists of a sequence of an isometry and its inverse. Hence, it does not require unitary inversion of $\map{U}$; isometry inversion is sufficient. The inverse of an isometry from a $d$-dimensional system to a $D$-dimensional system can be implemented with the same query cost as unitary inversion on a $d$-dimensional system \cite{Yoshida2025universal}. 

The staircase backstitch protocol can be extended inductively to general $K$-slot unitary superchannels.
\begin{theo}
    There exists a quantum circuit that implements an unknown $K$-slot unitary superchannel $\smap{U}$ by alternating between a call to the unitary staircase $\map{U}$ and a call to its inverse $\map{U}^{-1}$, invoking $\map{U}$ a total of $K+1$ times and $\map{U}^{-1}$ a total of $K$ times.
\end{theo}
\textit{Proof}.
Let $\smap{U}^k$ denote a $(K-k)$-slot unitary superchannel obtained by inserting the swap operation into the first $k$ slots of $\smap{U}$ for $k=0,\ldots,K$.
We have $\smap{U}^0 = \smap{U}$ and $\smap{U}^K = \map{U}$. Moreover, the staircase version of $\smap{U}^k$ is equal to $\map{U}$ for any $k$.

The circuit shown in Figure~\ref{fig:recursive_backstitch} implements a $K$-slot unitary superchannel $\smap{U}$ by sequentially invoking its staircase $\map{U}$, its inverse staircase $\map{U}^{-1}$, and the $(K-1)$-slot unitary superchannel $\smap{U}^1$.
Similarly, the $K-k$-slot superchannel $\smap{U}^k$ can be implemented using its staircase $\map{U}^k = \map{U}$, its inverse staircase $(\map{U}^k )^{-1} = \map{U}^{-1}$, and the $(K-k-1)$-slot superchannel $\smap{U}^{(k+1)}$.
By recursively applying this procedure starting from $k=0$, we obtain a circuit that realizes $\smap{U}^0 = \smap{U}$ by alternating calls to $\map{U}$, used a total of $K+1$ times, and to $\map{U}^{-1}$, used a total of $K$ times.
$\qed$
\begin{figure}[htbp]
    \centering
    \includegraphics[width=0.8\textwidth]{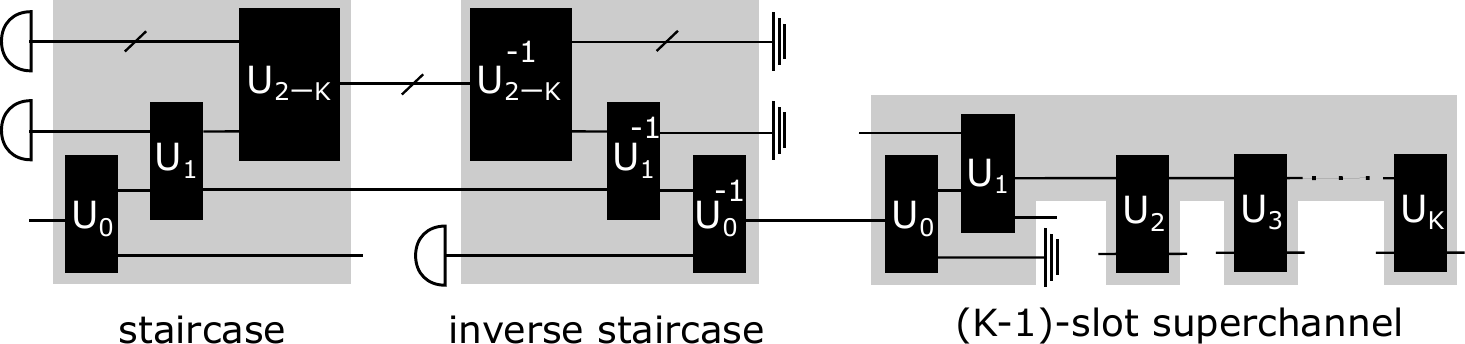}
    \caption{Circuit decomposition of a $K$-slot unitary superchannel $\smap{U}$ into its staircase $\map{U}$, its inverse $\map{U}^{-1}$, and the $(K-1)$-slot unitary superchannel $\smap{U}^1$ (as defined in the main text). While $\smap{U}$ consists of the unitaries $(U_0, U_1, \ldots ,U_K)$, the ``$U_{2-K}$'' black box denotes the staircase associated with $(U_2, \ldots ,U_K)$.}
    \label{fig:recursive_backstitch}
\end{figure}

It is remarkable that channel-to-superchannel conversion can be achieved deterministically with a finite number of channel queries, in contrast to state-to-channel conversion, which can be achieved only probabilistically with a finite number of states. By using a sufficient but constant number of unitary staircases, staircase backstitch reproduces the target unitary superchannel with certainty. This deterministic construction highlights a key advantage over partial teleportation, which, while capable of adding slots to a given unitary, is fundamentally limited by its low success probability.

A straightforward integration of staircase backstitch and channel pSAR \cite{sedlak2019PSAR} involves applying the backstitch to staircases recovered via the pSAR protocol. Because staircase backstitch succeeds deterministically, it imposes no overhead on the success probability, which remains governed by the scaling of the channel pSAR. As this probability tends to $1$ for $N \rightarrow \infty$, the integrated scheme eventually surpasses the partial teleportation method in the large-$N$ limit.

We do not explicitly derive the success probability of the integrated pSAR protocol here, as it depends on the specific method used to combine staircase backstitch and channel pSAR, leaving potential for further optimization. For instance, PBT—an optimal protocol for channel pSAR—could be tailored to retrieve the transposed staircase $\map{U}^\top$ without compromising the success probability. Subsequently, the inverse $\map{U}^{-1}$ can be obtained with fewer queries through the higher-order transformation of unitary complex conjugation, $\map{U} \mapsto \overline{\map{U}}$ \cite{miyazakiCC2019,ebler_optimal_2023}. In any case, the success probability is initially suppressed for small $N$, as both the forward and inverse channels must be invoked a finite number of times. However, it approaches unity in the asymptotic limit $N \rightarrow \infty$, consistent with the convergence of channel pSAR. The primary significance of staircase backstitch lies in proving that pSAR (and hence the retrospective intervention) can be achieved without any constant overhead on the success probability.

\begin{rema}
    The staircase backstitch protocol can be employed to demonstrate the feasibility of \textit{superchannel inversion}, a higher-order transformation that reverses the action of a given unitary superchannel. Specifically, given a finite number of calls to an unknown unitary superchannel, its inverse can be implemented deterministically and exactly.
    
    For a unitary superchannel $\smap{U}$ consisting of the sequence $(U_0, \ldots, U_K)$, its inverse $\smap{U}^{-1}$ is defined by the sequence of inverses $(U_K^{-1}, \ldots, U_0^{-1})$. Indeed, the staircase backstitch enables the realization of $\smap{U}^{-1}$ by invoking the inverse staircase $\map{U}^{-1}$ a total of $K+1$ times and the staircase $\map{U}$ a total of $K$ times. Again, both $\map{U}$ and $\map{U}^{-1}$ are obtainable from a finite number of calls to the original superchannel $\smap{U}$.
\end{rema}

\section{Maximum success probability}\label{sec:probability}
In this section, we outline the methodology and results of our analysis concerning the maximum success probabilities of pSAR for small-scale instances. While optimizing over general quantum circuits is challenging, the simplest $1$-to-$1$ pSAR can be analytically solved. The problem at hand can be reformulated as a semidefinite program (SDP) that is numerically tractable for small sizes. The analytical and numerical results are both achieved by utilizing the quantum comb formalism \cite{chiribella2008architechture,chiribella2008transforming,chiribella2009framework} in conjunction with group-theoretic techniques \cite{grinko2023gelfandtsetlinbasispartiallytransposed,Grinko_2024,grinko2025sequentialquantumprocessesgroup}.

\subsection{Comb formalism}
In the quantum comb formalism \cite{chiribella2008architechture,chiribella2008transforming,chiribella2009framework}, quantum superchannels are represented by matrices, allowing their compositions to be described through matrix operations. In the following, we provide a brief review of this formalism.

Let $\smap{N}$ be a quantum superchannel of type $(\dim \hil{0},\ldots,\dim \hil{2K+1})$, which contains only deterministic quantum channels. Let $C_{\smap{N}}$ be the Choi operator of $\smap{N}$ defined by
\begin{equation}
    C_{\smap{N}} := \map{N} \otimes_{k:even} \iden_{\hil{k}} \left( \otimes_{k:even} \Psi_{\hil{k}} \right) \in \lin(\hil{0} \otimes \cdots \otimes \hil{2K+1})
\end{equation}
where $\Psi_{\hil{k}} = \ketbra{\Psi_{\hil{k}}}$ on $\hil{k} \otimes \hil{k}$ is given by $\ket{\Psi_{\hil{k}
}} = \sum_i \ket{ii}$. This Choi operator satisfies the condition
\begin{subequations}\label{eq:comb_conditions}
    \begin{align}
    \label{eq:comb_condition1}  & \tr_{2k+1} C^{(k)} = \id_{2k} \otimes C^{(k-1)},\\
    \label{eq:comb_condition2}  & C^{(K-1)} := C_{\smap{N}}, \qquad C^{(1)} = 1,
\end{align}
\end{subequations}
where $C^{(k)}$ are recursively defined by \eqref{eq:comb_condition1}.
In contrast, any positive semi-definite operator $C$ acting on $\hil{0} \otimes \cdots \otimes \hil{2K+1}$ that satisfies the aforementioned conditions is the Choi operator of a deterministic quantum superchannel of type $(\dim \hil{0},\ldots,\dim \hil{2K+1})$. Such operators are formally referred to as deterministic quantum combs.

Furthermore, a positive semi-definite operator $C_p$ on the same space represents a probabilistic quantum superchannel (i.e., a component of a quantum instrument) if and only if there exists a deterministic quantum comb $C$ such that
\begin{equation}
    \label{eq:probabilistic_comb_condition} C_p \leq C.
\end{equation}
In this case, $C_p$ is termed a probabilistic quantum comb.

The composition of two quantum superchannels is represented by the link product. Let $\smap{N}_1$ and $\smap{N}_2$ be two quantum superchannels such that they can be connected at ports in the subset $J$. When we denote the composed superchannel by $\smap{N}_1 \circ_J \smap{N}_2$, the link product $\star$ is defined so that
\begin{equation}
    C_{\smap{N}_1 \circ_J \smap{N}_2} = C_{\smap{N}_1} \star C_{\smap{N}_2},
\end{equation}
holds.

\subsection{pSAR formalized by comb}
Note that the Choi operator $C_{\smap{N}}$ defined from a superchannel is mathematically indistinguishable from the operator $C_{\map{N}}$ defined from its corresponding staircase. Therefore, to analyze pSAR within the comb formalism, we must distinguish between these configurations by imposing specific constraints on the storage ($\smap{S}$) and retrieval ($\smap{R}$) operations.

To this end, it is convenient to reformulate the pSAR condition (Def.~\ref{def:pSAR}) in terms of an integrated SAR superchannel $\smap{L}$, formed by the composition of $\smap{S}$ and $\smap{R}$ via the memory system $\memory{}$.
The entire SAR superchannel $\smap{L}$ has the type
\begin{equation}
    \label{eq:type_L}   (1,[d_0,\ldots,d_{2K+1}]_N, 1, d_0,\ldots,d_{2K+1}).
\end{equation}
\begin{defi}\label{def:pSAR_comb}
    A probabilistic quantum comb $L$ of type \eqref{eq:type_L} is defined to realize pSAR of type $(d_0,\ldots,d_{2K+1})$ with probability $p$ when the following holds:
    \begin{equation}
        \label{eq:pSAR_comb} L \star C_{\smap{U}}^{\otimes N} = p ~C_{\smap{U}}, \qquad (^\forall \smap{U} \in \smap{U}[d_0,\ldots,d_{2K+1}]).
    \end{equation}
\end{defi}

The primary objective of this reformulation is to determine the maximum success probability, $p_\mathrm{max}$, by treating $L$ as the optimization variable.
The constraints \eqref{eq:comb_conditions} and \eqref{eq:probabilistic_comb_condition} defining the SAR superchannel type can be readily cast as an SDP.
However, the pSAR condition \eqref{eq:pSAR_comb} involves an infinite number of constraints, making it difficult to handle both analytically and numerically.
As described in the following subsection, we address this by leveraging the symmetry properties of the problem to transform the pSAR condition into a more tractable form for calculating $p_\mathrm{max}$.

\subsection{Unitary-equivalence on SAR combs}
Unitary-equivalence is a symmetry under the action of collective unitaries with their conjugates $U^{\otimes n} \otimes \bar{U}^{\otimes m}$. Operators that commute with this operator have a specific decomposition due to the mixed Schur-Weyl duality theorem.

The unitary equivalence in $L$ is observed in the following. Let the type of unitary superchannel be given by $(\dim \hil{0},\ldots,\dim \hil{2K+1})$. For $l = 0, \ldots ,2K+1$, let $P_l$ be the set of $N$ subsystems of $L$ that are connected to the $l$-th port of the superchannel. Additionally, let $\hil{l}^\mathrm{R}$ be the subsystem of $L$ corresponding to the $l$-th port of the retrieved unitary. For unitary $V_l$ on $\hil{l}$ define
\begin{equation}
    g_l(V):= \left( \otimes_{\hil{} \in P_l} V_{\hil{}} \right) \otimes \bar{V}_{\hil{l}^\mathrm{R}} \otimes \id,
\end{equation}
where the identity operator belongs to the remaining subsystems of $L$.
\begin{lemm}\label{lem:symmetry}
    Suppose that a probabilistic comb $L$ realizes the pSAR of unitary superchannels of type $[d_0,\ldots,d_{2K+1}]$ with probability $p$. Then there is a pair of deterministic comb $L^\mathrm{det}_\mathrm{sym}$ and probabilistic comb $L_\mathrm{sym} \leq L^\mathrm{det}_\mathrm{sym}$ that realizes the same pSAR with the same probability $p$, and additionally satisfies
    \begin{equation}
        \label{eq:commutation}  \left[ L^\mathrm{det}_\mathrm{sym} , g_l(V) \right]=0, \quad \left[ L_\mathrm{sym} , g_l(V) \right] =0 \qquad (\forall V: \text{unitary on } \hil{l}, \quad l=0,\ldots ,2K+1).
    \end{equation}
\end{lemm}
\textit{Proof sketch}.
$L_\mathrm{sym}$ and $L^\mathrm{det}_\mathrm{sym}$ can be constructed from the original $L$ and $L^\mathrm{det}$ by twirling. The operators $g_l(V) L g_l(V)^\dagger$ and $g_l(V) L^\mathrm{det} g_l(V)^\dagger$ can be shown to satisfy all the pSAR conditions with probability $p$, for any $V$ and $l$. Their Haar randomized versions
\begin{equation}
    L_{\mathrm{sym},l} := \int_{\mathrm{SU}(d_l)} dV g_l(V) L g_l(V)^\dagger, \qquad  L^\mathrm{det}_{\mathrm{sym},l} := \int_{\mathrm{SU}(d_l)} dV g_l(V) L^\mathrm{det} g_l(V)^\dagger,
\end{equation}
also realize pSAR with probability $p$. Additionally, the Haar randomized versions satisfy the commutation relations \eqref{eq:commutation} for $l$. We can thus construct $L_\mathrm{sym}$ and $L^\mathrm{det}_\mathrm{sym}$ by applying the Haar randomizations to $l=0,\ldots,2K+1$.
$\qed$
\begin{lemm}\label{lem:identity_is_enough}
    Suppose $K=1$. Let $L$ be a probabilistic comb $L$ of the same type as the SAR comb. If $L$ satisfies the commutation relation
    \begin{equation}
        \label{eq:commutation2}  \left[ L , g_l(V) \right]=0, \qquad (\forall V: \text{unitary on } \hil{l}, \quad l=0,\ldots ,2K+1),
    \end{equation}
    then it realizes pSAR with probability $p$ if and only if
    \begin{equation}
        \label{eq:pSAR_comb2} L \star C_{\widetilde{\iden}}^{\otimes N} = p ~C_{\widetilde{\iden}},
    \end{equation}
    where $\widetilde{\iden}$ is the identity superchannel in $\smap{U}[d_0,\ldots,d_{3}]$.
\end{lemm}
\noindent See Appendix~\ref{subsec:proof1} for the full proof of Lemma \ref{lem:symmetry} and \ref{subsec:proof2} for the proof of Lemma \ref{lem:identity_is_enough}.

The mixed Schur-Weyl duality, as detailed in \cite{Grinko_2024, grinko2023gelfandtsetlinbasispartiallytransposed}, implies that $L$ has the symmetry \eqref{eq:commutation2} if and only if it has the following decomposition:\footnote{The notation in Eq.\,\eqref{eq:mixed_SW}, such as 
\(\left( \sum_{a_l \in A_l} \right)_{l=0,\dots,2K+1}\), 
compactly represents the nested summation \(\sum_{a_0 \in A_0} \dots \sum_{a_{2K+1} \in A_{2K+1}}\).}
\begin{equation}
    \label{eq:mixed_SW} L = \left( \sum_{\lambda_l \in \hat{\maptp}^{d_l}_{N,1}} \sum_{S_l,T_l \in \mathrm{Paths}(\lambda_l)} \right)_{l=0,\dots,2K+1} 
    c_{(S_l,T_l)_l}^{(\lambda_l)_l} \bigotimes_{l=0}^{2K+1} E_{S_l, T_l}^{\lambda_l}.
\end{equation}
Here, \( \hat{\maptp}^{d_l}_{N,1} \) denotes the set of irreducible representations, labeled by mixed Young diagrams \(\lambda_l\), of the matrix algebra \( \maptp^{d_l}_{N,1} \) of partially transposed permutations. The elements \( E_{S_l, T_l}^{\lambda_l} \in \opsp{\bigotimes_{n=1}^{N+1} \lnsp{H}_l^n} \) are matrix units of \( \maptp^{d_l}_{N,1} \). The coefficients \( c_{(S_l,T_l)_l}^{(\lambda_l)_l} \) are variables that depend on \( L \).\footnote{Notations such as \(\maptp^{d_l}_{N,1}\) and \(\mathrm{Paths}(\lambda_l)\) follow \cite{grinko2023gelfandtsetlinbasispartiallytransposed}.}

In summary, to optimize the pSAR comb $L$ for $K=1$-slot superchannels and determine the maximum success probability, we focus on the constraints \eqref{eq:comb_conditions}, \eqref{eq:probabilistic_comb_condition}, and \eqref{eq:pSAR_comb2} within the decomposition framework of \eqref{eq:mixed_SW}.
For multi-slot combs ($K \geq 2$) these constraints remain necessary but may not be sufficient since Lemma \ref{lem:identity_is_enough} no longer holds. Therefore, the numerically obtained value of the probability under these constraints is only an upper bound of the maximum success probability.

\subsection{Optimal 1-to-1 pSAR}
The simplest case, $1$-to-$1$ pSAR, admits an analytical solution as follows.
\begin{theo}\label{thm:n1}
    The maximum success probability for the $1$-to-$1$ pSAR of staircases and superchannels of type $(d_0,\ldots,d_{2K+1})$ is given by
    \begin{equation}
        p_{\mathrm{max},1}^{\map{U} \rightarrow \map{U}} = p_{\mathrm{max},1}^{\smap{U} \rightarrow \map{U}} = p_{\mathrm{max},1}^{\smap{U} \rightarrow \smap{U}} = \Pi_{k=0,\mathrm{even}}^{2K} \frac{1}{d_k^2}.
    \end{equation}
    This is achieved by independently applying probabilistic teleportation to the systems $\hil{0},\hil{2},\ldots,\hil{2K}$.
\end{theo}
For the superchannel-to-superchannel pSAR, the optimal protocol described in Theorem~\ref{thm:n1} coincides with Protocol 1, which is based on partial teleportation (see Section~\ref{subsec:partial_teleportation}). Note that probabilistic port-based teleportation reduces to standard teleportation in the absence of ancillary ports, i.e., when $N=1$.

The theorem is a consequence of the following lemma concerning the $1$-to-$1$ pSAR of multiple independent unitary channels.
\begin{lemm}\label{lem:product_rule}
    The maximum success probability of retrieving a tensor-product unitary channel $\otimes_{i=1}^I \map{U}_i$ $(U_i \in U(d_i))$ from any storage scheme that invokes each independent channel exactly once at arbitrary times is given by
    \begin{equation}
        \prod_{i=1}^I \frac{1}{d_i^2}.
    \end{equation}
\end{lemm}
This lemma implies that the optimal $1$-to-$1$ pSAR for tensor-product channels is realized by independently applying probabilistic teleportation to each component—the same protocol that is optimal for individual quantum channels. This ``product rule'' is a known feature in the estimation of tensor-product channels; specifically, the maximization of certain figures of merit is achieved through independent estimation of the individual components \cite{Chiribella_2012}. Here, we demonstrate that a similar product rule holds for the pSAR of quantum channels. Our proof of Lemma~\ref{lem:product_rule}, presented in Appendix~\ref{sec:product_rule}, employs a methodology distinct from the approach in Ref.~\cite{Chiribella_2012}.

\textit{Proof of Theorem \ref{thm:n1}}.
Let $(d_0,\ldots,d_{2K+1})$ be the type of a unitary superchannel, and let
\begin{equation}
    d_k = p_{k,1} p_{k,2} \cdots p_{k, m_k}, \qquad (k=0,\ldots,2K+1)
\end{equation}
be the factorization of $d_k$ into prime numbers (where $p_{k,i}$ and $p_{k,i'}$ may be equal for $i \neq i'$). Each Hilbert space admits a corresponding decomposition $\hil{k} = \hil{k,1} \otimes \cdots \otimes \hil{k,m_k}$.

Since the superchannel is unitary, for each $(k,i)$ with even $k$ (input) there must exist an index $(k',i')$ with odd $k'$ (output) such that $k < k'$ and $p_{k,i} = p_{k',i'}$. It is possible to make a one-to-one matching by such pairing between indices belonging to input ports and those belonging to output ports. We denote $(k,i) \rightarrow (k',i')$ when $(k,i)$ and $(k',i')$ are a pair of this matching.

Within the class of unitary superchannels of type $(d_0,\ldots,d_{2K+1})$, we consider those that decompose into independent unitary channels from $\hil{k,i}$ to $\hil{k',i'}$ for all pairs $(k,i) \rightarrow (k',i')$ in the matching. According to Lemma~\ref{lem:product_rule}, the maximum success probability for retrieving these restricted unitary superchannels is
\begin{equation}
    \label{eq:upper_bound}  \prod_{(k,i),k:even} \frac{1}{p_{k,i}^2} = \prod_{k=0,even}^{2K} \frac{1}{d_k^2}.
\end{equation}
By construction, this serves as an upper bound for the success probability of $1$-to-$1$ superchannel-to-staircase pSAR.

This upper bound \eqref{eq:upper_bound} is achievable by applying probabilistic teleportation independently to each subsystem $\hil{k,i}$. Thus, \eqref{eq:upper_bound} gives the maximum success probability $p_{\mathrm{max},1}^{\smap{U} \rightarrow \map{U}} (d_0,\ldots,d_{2K+1})$. Furthermore, since this independent probabilistic teleportation protocol is also applicable to both staircase-to-staircase and superchannel-to-superchannel pSAR, we have
\begin{equation}
    p_{\mathrm{max},1}^{\smap{U} \rightarrow \map{U}} = \prod_{k=0,even}^{2K} \frac{1}{d_k^2} \leq p_{\mathrm{max},1}^{\map{U} \rightarrow \map{U}}, p_{\mathrm{max},1}^{\smap{U} \rightarrow \smap{U}}.
\end{equation}
These inequalities must be equalities, as the general relations \eqref{eq:inequality1} and \eqref{eq:inequality2} hold.
$\qed$\\

\subsection{Numerical analysis}
All constraints for the pSAR comb $L$, namely \eqref{eq:comb_conditions}, \eqref{eq:probabilistic_comb_condition}, and \eqref{eq:pSAR_comb2} within the decomposition framework of \eqref{eq:mixed_SW}, can be cast as an SDP.
However, the high dimensionality of the variables necessitates additional techniques to reduce the computational complexity. The details of this reduction process are provided in Appendix~\ref{sec:reduction}.

We summarize the numerically obtained maximum success probabilities for the three configurations of pSAR as follows: staircase pSAR in Table~\ref{tab:staircase}, superchannel-to-staircase pSAR in Table~\ref{tab:superchannel-to-staircase}, and superchannel pSAR in Table \ref{tab:superchannel}. Note that the numerical value is only an upper bound for the type $(4, 2,2,2,2,4)$, as Lemma \ref{lem:identity_is_enough} is no longer available for $K \geq 2$. We used MOSEK as the solver and CVX and the interpreter in MATLAB to numerically solve the SDP.\footnote{The code is available at \url{https://github.com/butterfly1026/psar_of_pure_combs}}
\begin{table}[htbp]
    \centering
    \begin{tabular}{|c|c|S[table-format=1.6,group-separator={}]|S[table-format=1.6,group-separator={}]|c|}
    \hline
    \( N \) & {staircase type} & $p^{\map{U} \rightarrow \map{U}}_{\mathrm{SDPmax},N}$ & { \( p_N^{\mathrm{PBT}} \) } & $(p^{\map{U} \rightarrow \map{U}}_{\mathrm{SDPmax},N} - p_N^{\mathrm{PBT}})/p_N^{\mathrm{PBT}}$ \\
    \hline
    1 & (4, 2, 2, 4) & 0.01563  & 0.01563  & \( < 10^{-10} \) \\
    1 & (4, 2, 2, 2, 2, 4) & 0.003906 & 0.003906 & \( < 10^{-10} \) \\
    \hline
    2 & (4, 2, 2, 4) & 0.03077 & 0.03077 & \( 1.5 \times 10^{-4} \) \\
    2 & (4, 2, 3, 6) & 0.01380 & 0.01379 & \( 4.1 \times 10^{-4} \) \\
    2 & (6, 2, 2, 6) & 0.01380 & 0.01379 & \( 4.2 \times 10^{-4} \) \\
    2 & (6, 3, 2, 4) & 0.01380 & 0.01379 & \( 4.1 \times 10^{-4} \) \\
    2 & (6, 3, 3, 6) & 0.006161 & 0.006154 & \( 1.2 \times 10^{-3} \) \\
    2 & (4, 2, 2, 2, 2, 4) & 0.007999 & 0.007782 & \( 2.7 \times 10^{-2} \) \\
    \hline
    3 & (4, 2, 2, 4) & 0.04723 & 0.04545 & \( 3.9 \times 10^{-2} \) \\
    \hline
    \end{tabular}
    \caption{
        The maximum success probabilities of $N$-to-one pSAR of unitary staircases, numerically obtained by solving SDP. The numerically maximized success probabilities $p^{\map{U} \rightarrow \map{U}}_{\mathrm{SDPmax},N}$ are compared to the success probability $p^\text{PBT}_N$ of the PBT protocol.
    }
\label{tab:staircase}
\end{table}
\begin{table}[htbp]
    \centering
    \begin{tabular}{|c|c|S[table-format=1.6,group-separator={}]|S[table-format=1.6,group-separator={}]|c|}
    \hline
    \( N \) & {staircase type} & $p^{\smap{U} \rightarrow \map{U}}_{\mathrm{SDPmax},N}$ & { \( p_N^{\mathrm{PBT}} \) } & $(p^{\smap{U} \rightarrow \map{U}}_{\mathrm{SDPmax},N} - p_N^{\mathrm{PBT}})/p_N^{\mathrm{PBT}}$ \\
    \hline
    1 & (4, 2, 2, 4) & 0.01563  & 0.01563  & \( < 10^{-10} \) \\
    1 & (4, 2, 2, 2, 2, 4) & 0.003906 & 0.003906 & \( < 10^{-10} \) \\
    \hline
    2 & (4, 2, 2, 4) & 0.03077 & 0.03077 & \( 1.5 \times 10^{-4} \) \\
    2 & (4, 2, 3, 6) & 0.01380 & 0.01379 & \( 4.1 \times 10^{-4} \) \\
    2 & (6, 2, 2, 6) & 0.01380 & 0.01379 & \( 4.2 \times 10^{-4} \) \\
    2 & (6, 3, 2, 4) & 0.01380 & 0.01379 & \( 4.1 \times 10^{-4} \) \\
    2 & (6, 3, 3, 6) & 0.006164 & 0.006154 & \( 1.6 \times 10^{-3} \) \\
    2 & (4, 2, 2, 2, 2, 4) & 0.007952 & 0.007782 & \( 2.2 \times 10^{-2} \) \\
    \hline
    3 & (4, 2, 2, 4) & 0.04748 & 0.04545 & \( 4.4 \times 10^{-2} \) \\
    \hline
    \end{tabular}
    \caption{
        The maximum success probabilities of $N$-to-one superchannel-to-staircase pSAR, numerically obtained by solving SDP. The numerically maximized success probabilities $p^{\smap{U} \rightarrow \map{U}}_{\mathrm{SDPmax},N}$ are compared to the success probability $p^\text{PBT}_N$ of the PBT protocol.
    }
\label{tab:superchannel-to-staircase}
\end{table}
\begin{table}[htbp]
    \centering
    \begin{tabular}{|c|c|c|c|c|}   
    \hline
    \( N \) & superchannel type & { \( p^{\smap{U} \rightarrow \smap{U}}_{\mathrm{SDPmax},N} \) } & { \( p^\text{tele}_N \) } & \( (p^{\smap{U} \rightarrow \smap{U}}_{\mathrm{SDPmax},N} - p^\text{tele}_N)/ p^\text{tele}_N \) \\
    \hline
    1 & (4, 2, 2, 4) & 0.01563  & 0.01563  & \( < 10^{-10} \) \\
    1 & (4, 2, 2, 2, 2, 4) & 0.003906 & 0.003906 & \( < 10^{-10} \) \\
    \hline
    2 & (4, 2, 2, 4) & 0.02942 & 0.02941 & \( 1.5 \times 10^{-4} \) \\
    2 & (4, 2, 3, 6) & 0.01308 & 0.01307 & \( 5.5 \times 10^{-4} \) \\
    2 & (6, 2, 2, 6) & 0.01352 & 0.01351 & \( 4.3 \times 10^{-4} \) \\
    2 & (6, 3, 2, 4) & 0.01352 & 0.01351 & \( 4.3 \times 10^{-4} \) \\
    2 & (6, 3, 3, 6) & 0.006014 & 0.006006 & \( 1.3 \times 10^{-3} \) \\
    2 & (4, 2, 2, 2, 2, 4) & 0.007592 & 0.007353 & \( 3.2 \times 10^{-2} \) \\
    \hline
    3 & (4, 2, 2, 4) & 0.04329 & 0.04167 & \( 3.9 \times 10^{-2} \) \\
    \hline
    \end{tabular}
    \caption{The maximum success probabilities for the $N$-to-one pSAR of unitary superchannels, numerically obtained by solving SDP. The numerically maximized success probabilities $p^{\smap{U} \rightarrow \smap{U}}_{\mathrm{SDPmax},N}$ are compared to the success probability $p^\text{tele}_N$ of the partial teleportation protocol.}
\label{tab:superchannel}
\end{table}

For both staircase and superchannel-to-staircase pSARs, the numerical results are compared with the success probabilities $p_N^{\mathrm{PBT}} := {N}/{(N-1+D^2)}$ (where $D := d_0 \times \cdots \times d_{2K}$) achieved by the PBT protocol, assuming collective PBT on all input ports. For superchannel pSAR, the values are compared with the success probabilities $p_N^{\mathrm{tele}}$ (given by \eqref{eq:p_teleportation}) attained by the partial teleportation protocol.
The results at $N=1$ confirm Theorem \ref{thm:n1}, namely, that the maximum success probability of $1$-to-$1$ pSAR is achieved by the probabilistic teleportation.

To obtain qualitative references for numerical errors, we also perform the calculations for unitary staircases and sueprchannels that reduce to products of unitary channels.
These staircases and superchannels are of type $(a,a,b,b,\ldots,c,c)$ when their component unitary channels are defined on spaces of dimensions $a,b,\ldots,c$.
In these cases, PBT can be applied independently to each component channel.
The results for staircase and superchannel pSAR are given in Appendix \ref{sec:tensor_product_numerical}. 
For $N \geq 2$, the optimality of this tensor PBT protocol has not been proven, but the numerical results are close to the tensor PBT values, and we use these results as references.
The relative differences from the tensor PBT values are at most $1.4 \times 10^{-3}$ for staircase pSAR and $7.0 \times 10^{-4}$ for superchannel pSAR.
The SDPs are solved using finite-precision arithmetic, so errors in the equality and positive-semidefinite constraints may affect the computed objective values.

For $N=2$ and $K=1$, the relative differences in Tables \ref{tab:staircase}, \ref{tab:superchannel-to-staircase}, and \ref{tab:superchannel} range from $1.5 \times 10^{-4}$ to $1.6 \times 10^{-3}$.
They are similar in size to those in the tensor-product cases with $N=2$.
For fixed $N$ and $K$, the protocol-based value decreases as the port dimensions increase, so the same absolute difference gives a larger relative difference.
Increasing the port dimensions also increases the size of the reduced SDP.
Thus, the numerical results are consistent with the compared protocols being optimal in this regime, although they do not prove their optimality.


For $N=2$ and type $(4,2,2,2,2,4)$, the relative differences range from $2.2 \times 10^{-2}$ to $3.2 \times 10^{-2}$.
In this case, the dependence on an intermediate unitary remains after the reduction, and the constraints used in the SDP are necessary but may not be sufficient.
As noted above, the resulting values are only upper bounds and do not determine the optimal success probabilities.

For $N=3$ and type $(4,2,2,4)$, the relative differences range from $3.9 \times 10^{-2}$ to $4.4 \times 10^{-2}$.
The reduced SDPs are larger than those for $N=2$, and the differences are also much larger than those in the tensor-product cases with $N=3$.
Therefore, these differences may be caused by numerical errors, protocols with higher success probabilities than the compared protocols, or both.

Recall that superchannel pSAR enables retrospective intervention in the superchannel, a capability that superchannel-to-staircase pSAR lacks. 
For all rows with $N \geq 2$ in Tables \ref{tab:superchannel-to-staircase} and \ref{tab:superchannel}, the numerical value for superchannel-to-staircase pSAR is larger than that for superchannel pSAR.
This suggests that retrospective intervention has an additional cost, without assuming that the compared protocols are optimal.

However, as previously noted, partial teleportation is proven to be sub-optimal for large $N$. Specifically, it cannot achieve a unit success probability, regardless of the number of queries to the unknown unitary superchannel. In contrast, the protocol based on staircase backstitch attains deterministic success in the asymptotic limit. The scope of our current numerical analysis was insufficient to observe a significant divergence between $p^\text{tele}_N$ and the true maximum success probability.

\section{Conclusion}\label{sec:conclusion}
In this work, we have investigated the probabilistic storage and retrieval (pSAR) of unitary staircases and unitary superchannels using both analytical and numerical approaches. We proposed pSAR protocols for unitary superchannels by decomposing the task into two stages: the state-channel transformation, as represented by existing channel pSAR protocols, and a subsequent channel-to-superchannel transformation. We identified partial teleportation and staircase backstitch as physical implementations of the channel-to-superchannel transformation, which in turn lead to two distinct pSAR strategies. As a byproduct of our analysis of pSAR for unitary superchannels, we also present a protocol for universally inverting a unitary superchannel while preserving its intervention structure.

Our numerical results within the regime of $K=1$ and $N \leq 2$ indicate that the maximum success probability for the pSAR of unitary staircases is close to the probability achieved by port-based teleportation (PBT), the known optimal protocol for unitary channel pSAR \cite{sedlak2019PSAR}, and are consistent with PBT being optimal. This finding aligns with the lineage of studies on the SAR of restricted unitary classes, such as phase gates \cite{sedlak2020PSARphasegate} and pairs of unitaries \cite{sedlak2024twounitaries}. However, in contrast to those cases, the restriction to unitary staircases does not appear to enhance the pSAR success probability within the examined numerical range.

In the same regime, the maximum success probability for the $N$-to-one pSAR of unitary superchannels is closely matched by that of the partial-teleportation protocol. This result confirms that PBT---while optimal for channel pSAR---is insufficient for superchannel pSAR. By contrasting this with our superchannel-to-staircase pSAR results, 
we observe numerical differences that suggest an additional cost of retrospective intervention for small $N$, without assuming that the compared protocols are optimal.

However, the close agreement with partial teleportation for small $N$ is likely not universal. Theoretical considerations show that the staircase backstitch protocol outperforms partial teleportation for sufficiently large $N$. Our construction of the staircase backstitch predicts a transition in the optimal pSAR strategy as $N$ increases, although the current numerical analysis was insufficient to observe this crossover directly. Synthesizing these findings, an intriguing open question remains: how does the partial teleportation protocol, which agrees closely with our numerical results for small $N$, connect to a regime for larger $N$ in which a constant overhead in the success probability can be avoided? It is conceivable that either the staircase backstitch becomes optimal in the asymptotic limit, or there exists a yet-undiscovered variant of PBT capable of coordinating delays across multiple slots at different times. Future numerical investigations with higher stability and precision for larger $N$ may reveal this transition.

Another open question is the existence of pSAR protocols for non-unitary quantum superchannels that achieve unit success probability asymptotically. While one might intuitively consider identifying a superchannel via infinite queries and reconstructing it, such a procedure approaches deterministic exact SAR via approximate SAR, which is conceptually distinct from our probabilistic exact SAR approach. In this context, extending the staircase backstitch to general superchannels is a potential direction. 

Finally, as this study focused on superchannels with a definite causal order, we did not consider indefinite causal structures \cite{oreshkov2012quantum, araujo2017purification} such as the quantum switch \cite{chiribella2013quantum}. Generalizing the pSAR framework to incorporate indefinite causal order represents a compelling avenue for future research.

\paragraph{Acknowledgment}
We would like to thank Alastair Abbott, Jessica Bavaresco, Dmitry Grinko, Marco T\'{u}lio Quintino, Akihito Soeda, and Satoshi Yoshida for helpful discussions. This work was supported by MEXT Quantum Leap Flagship Program (MEXT QLEAP) JPMXS0118069605 and JPMXS0120351339; Japan Science and Technology Agency (JST) as part of Adopting Sustainable Partnerships for Innovative Research Ecosystem (ASPIRE), Grant Number JPMJAP25A3; JST CREST, Grant Number JPMJCR25I5; JST NEXUS, Grant Number JPMJNX26C9; JSPS KAKENHI Grant No. 21H03394 and No. 23K21643; and IBM Quantum.

\bibliographystyle{unsrt}
\bibliography{references}

\appendix

\section{Probabilistic port-based teleportation for channel pSAR}\label{sec:PBT}
We incorporate probabilistic PBT \cite{ishizaka2008asymptotic,ishizaka2009quantum,studzinski2017port,mozrzymas2018optimal} as a subroutine for superchannel pSAR. We review its application to channel pSAR, specifically focusing on the variant that utilizes multiple maximally entangled states as a reference.

For an $N$-port probabilistic PBT, the sender and the receiver share $N$ maximally entangled states $\Psi_{\hil{i}}$ on $\hil{i}^A \otimes \hil{i}^B$ where $\hil{i} \sim \hil{}$ for $i=1,\ldots,N$. To transmit an unknown quantum state $\rho$ on $\hil{}$, the sender performs a joint POVM $\{ P_i \}_{i=0,1,\ldots,N}$ on the space $\otimes_i \hil{i}^A \otimes \hil{}$. If the outcome is $i=0$, the teleportation fails. If the outcome is $i \in \{ 1,\ldots,N \}$, the unnormalized state at the $i$-th port of the receiver is given by
\begin{equation}
    \tr_{\overline{\hil{i}^B}} [P_i (\otimes_j \Psi_{\hil{j}}) \otimes \rho] = \frac{1}{N-1+\dim \hil{}^2} \rho.
\end{equation}

To adapt probabilistic PBT for the pSAR of channels, $N$ instances of an unknown channel $\map{C}$ are applied to the receiver's side of the maximally entangled state during the storage stage (See Figure~\ref{fig:pbt}). In the retrieval stage, once the input state $\rho$ of the channel is provided, the joint POVM is performed. For a successful measurement outcome $i$, the state at the $i$th port becomes
\begin{equation}
    \tr_{\overline{\hil{i}^B}} [P_i (\otimes_j (\map{C} \otimes \iden) \Psi_{\hil{j}}) \otimes \rho] = \map{C} \left( \tr_{\overline{\hil{i}^B}} [P_i (\otimes_j \Psi_{\hil{j}}) \otimes \rho] \right) = \frac{1}{N-1+\dim \hil{}^2} \map{C}(\rho).
\end{equation}
The retrieval is completed by selecting the $i$th port corresponding to the measurement result. The total success probability is the probability of obtaining any outcome $i \in \{ 1,\ldots,N \}$, which is given by $N/(N-1+\dim \hil{}^2)$.
\begin{figure}[htb]
    \centering
    \includegraphics[width=0.4\textwidth]{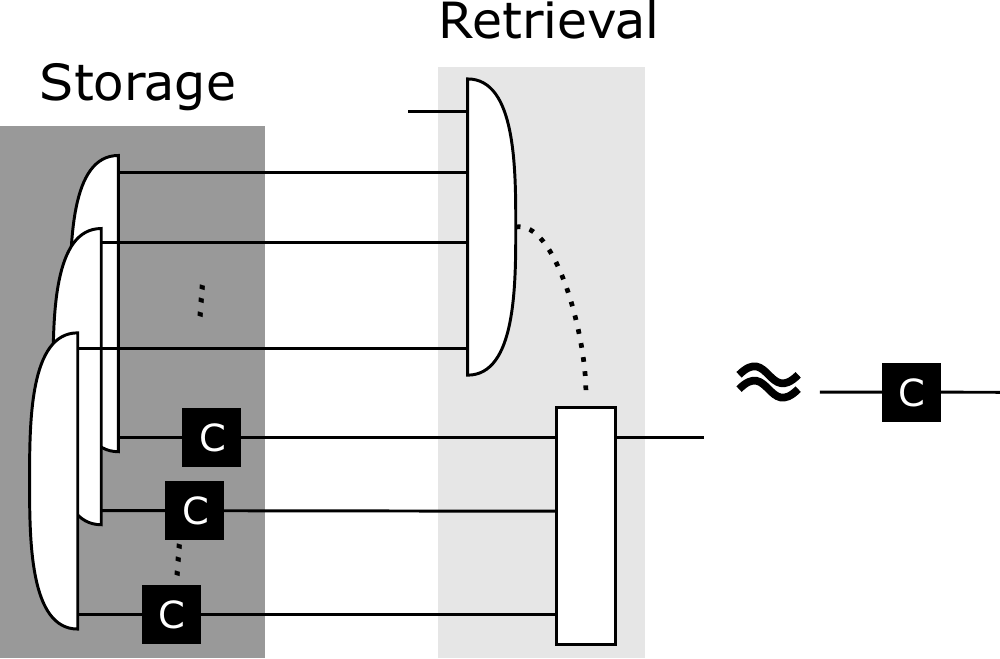}
    \caption{An optimal channel pSAR protocol implemented by PBT.}
    \label{fig:pbt}
\end{figure}

\section{Symmetry of SAR combs}
\subsection{Proof of lemma \ref{lem:symmetry}}\label{subsec:proof1}
Let $L$ be a probabilistic comb that satisfies the pSAR condition \eqref{eq:pSAR_comb}, namely,
\begin{equation}
    \label{eq:pSAR_comb_appendix} L \star C_{\smap{U}}^{\otimes N} = p ~C_{\smap{U}}, \qquad (^\forall \smap{U} \in \smap{U}[d_0,\ldots,d_{2K+1}]),
\end{equation}
and $L^\mathrm{det}$ be the deterministic comb of the same type that satisfies $L \leq L^\mathrm{det}$.
For any $l$ and any unitary operator $V$ on $\hil{l}$, the operators $g_l(V) L g_l(V)^\dagger$ and $g_l(V) L^\mathrm{det} g_l(V)^\dagger$ are probabilistic and deterministic combs of the same type, since they are local-unitary equivalent to the original ones, respectively. The unitary transformation also preserves the order, and thus we have $g_l(V) L g_l(V)^\dagger \leq g_l(V) L^\mathrm{det} g_l(V)^\dagger$.

To check that $g_l(V) L g_l(V)^\dagger$ realizes pSAR with probability $p$, we consider the cases of odd $l$ and even $l$, separately.\\
Case (i): $l$ is even. In this case $\otimes_{\hil{} \in P_l} V_{\hil{}}$ and $\bar{V}_{\hil{l}^\mathrm{R}}$ operate on several output and input ports of $L$, respectively. We omit identity operators in the following deduction. We have
\begin{align*}
    & g_l(V) L g_l(V)^\dagger \star C_{\smap{U}}^{\otimes N}\\
    &= \bar{V}_{\hil{l}^\mathrm{R}} \left( \otimes_{\hil{} \in P_l} V_{\hil{}} \right) L \left( \otimes_{\hil{} \in P_l} V_{\hil{}} \right) \star C_{\smap{U}}^{\otimes N} V_{\hil{l}^\mathrm{R}}^\top \\
    &= \bar{V}_{\hil{l}^\mathrm{R}} L \star C_{\smap{U} \circ \map{V}}^{\otimes N} V_{\hil{l}^\mathrm{R}}^\top \\
    &= p~ \bar{V}_{\hil{l}^\mathrm{R}} C_{\smap{U} \circ \map{V}} V_{\hil{l}^\mathrm{R}}^\top \\
    &= p~ C_{\smap{U}},
\end{align*}
as required, where $\smap{U} \circ \map{V}$ stands for the unitary superchannel such that the unitary channel $\map{V}(\cdot) = V \cdot V^\dagger$ operates at port $l$ before $\smap{V}$.\\
Case (ii): $l$ is odd. In this case $\otimes_{\hil{} \in P_l} V_{\hil{}}$ and $\bar{V}_{\hil{l}^\mathrm{R}}$ operate on several input and output ports of $L$, respectively. We have
\begin{align*}
    & g_l(V) L g_l(V)^\dagger \star C_{\smap{U}}^{\otimes N}\\
    &= \bar{V}_{\hil{l}^\mathrm{R}} \left( \otimes_{\hil{} \in P_l} V_{\hil{}} \right) L \left( \otimes_{\hil{} \in P_l} V_{\hil{}} \right) \star C_{\smap{U}}^{\otimes N} V_{\hil{l}^\mathrm{R}}^\top \\
    &= \bar{V}_{\hil{l}^\mathrm{R}} L \star C_{\map{V}^\top \circ \smap{U}}^{\otimes N} V_{\hil{l}^\mathrm{R}}^\top \\
    &= p~ \bar{V}_{\hil{l}^\mathrm{R}} C_{\map{V}^\top \circ \smap{U}} V_{\hil{l}^\mathrm{R}}^\top \\
    &= p~ C_{\smap{U}},
\end{align*}
as required, where $\map{V}^\top \circ \smap{U}$ stands for the unitary superchannel such that the unitary channel $\map{V}^\top(\cdot) = V^\top \cdot \bar{V}$ operates at port $l$ after $\smap{V}$.

Since $g_l(V) L g_l(V)^\dagger$ is a valid SAR comb and realizes the pSAR with probability $p$ for any unitary $V$, so does its Haar randomization
\begin{equation}
    L_{\mathrm{sym},l} := \int_{\mathrm{SU}(d_l)} dV g_l(V) L g_l(V)^\dagger.
\end{equation}
This probabilistic comb satisfies the relation $L_{\mathrm{sym},l} \leq L^\mathrm{det}_{\mathrm{sym},l}$ with the deterministic comb
\begin{equation}
     L^\mathrm{det}_{\mathrm{sym},l} := \int_{\mathrm{SU}(d_l)} dV g_l(V) L^\mathrm{det} g_l(V)^\dagger,
\end{equation}
since each summand satisfies this property. They also satisfy the commutation relation $[L_{\mathrm{sym},l}, g_l(V)]=0$ and $[L^\mathrm{det}_{\mathrm{sym},l}, g_l(V)] = 0$ by construction.
Therefore, applying Haar randomization to all ports $l=0,\ldots, 2K+1$, we obtain a SAR comb with the desired property.

\subsection{Proof of Lemma \ref{lem:identity_is_enough}}\label{subsec:proof2}
We focus on the sufficiency (the ``if'' direction) since the ``only if'' direction is trivial.

Let the unknown unitary superchannel be of type $(\dim \hil{0},\dim \hil{1},\dim \hil{2},\dim \hil{3})$. The memory system $\memory{}$ of the superchannel has dimension $d_{\memory{}} = \dim \hil{0} / \dim \hil{1} = \dim \hil{3} / \dim \hil{2}$. The unitary superchannel $\smap{U}$ of this type can be uniquely constructed by a pair of unitary operators $U_0$ on $\hil{0}$ and $U_3$ on $\hil{3}$.
In the comb representation, we have
\begin{equation}
    C_{\smap{U}} = C_{\map{U}_3} \star C_{\widetilde{\iden}} \star C_{\map{U}_0}.
\end{equation}
The action of the comb $L$ on $C_{\smap{U}}$ is given by
\begin{align}
    L \star C_{\smap{U}}^{\otimes N} &= L \star (C_{\map{U}_3})^{\otimes N} \star (C_{\map{U}_0})^{\otimes N} \star (C_{\widetilde{\iden}})^{\otimes N} \\
    \label{eq:appx1} &= (U_3^\top \otimes U_0)^{\otimes N} L (U_3^\top \otimes U_0)^{\dagger \otimes N} \star (C_{\widetilde{\iden}})^{\otimes N} \\
    \label{eq:appx2} &= (U_3 \otimes U_0^\top)_{\hil{}^\mathrm{R}} L (U_3 \otimes U_0^\top )^\dagger_{\hil{}^\mathrm{R}} \star (C_{\widetilde{\iden}})^{\otimes N} \\
    &= (U_3 \otimes U_0^\top)_{\hil{}^\mathrm{R}} L \star (C_{\widetilde{\iden}})^{\otimes N} (U_3 \otimes U_0^\top )^\dagger_{\hil{}^\mathrm{R}} \\
    &= (U_3 \otimes U_0^\top)_{\hil{}^\mathrm{R}} C_{\widetilde{\iden}} (U_3 \otimes U_0^\top )^\dagger_{\hil{}^\mathrm{R}} \\
    &= C_{\smap{U}},
\end{align}
which proves the lemma.
To obtain Eq.~\eqref{eq:appx2} from Eq.~\eqref{eq:appx1}, observe that the symmetry condition \eqref{eq:commutation2} can be rewritten as $g_l(V) L g_l(V)^\dagger =L$, and thus to
\begin{equation}
    \left( \otimes_{\hil{} \in P_l} V_{\hil{}} \right) L \left( \otimes_{\hil{} \in P_l} V_{\hil{}} \right)^\dagger = V^\top_{\hil{l}^\mathrm{R}} L \bar{V}_{\hil{l}^\mathrm{R}},
\end{equation}
where the identity operators on other ports are omitted for brevity.

\section{Product rule for pSAR of quantum channels}\label{sec:product_rule}
In this section, we present the proof of Lemma~\ref{lem:product_rule}, which states that the maximum success probability for $1$-to-$1$ pSAR of tensor-product channels is attained by applying independent probabilistic teleportations.

Let $L$ be a probabilistic comb realizing the $1$-to-$1$ pSAR of product of unitary channels $\left( \map{U}_i \right)_{i =1,\ldots,I }$ where $ \map{U}_i  \in \map{U}[d_i,d_i] $, and let $L^\mathrm{det}$ be a deterministic comb such that $L \leq L^\mathrm{det}$. We assume an arbitrary temporal structure for the storage phase, meaning that no constraints are imposed on the relative ordering of the inputs and outputs of each black box. The retrieval phase is intended to reproduce the parallel tensor product $\otimes_i \map{U}_i$ (an assumption invoked only at the final stage of the proof). We denote the input and output spaces of the stored unitary channels $\map{U}_i$ as $\hil{i}^\mathrm{S,in}$ and $\hil{i}^\mathrm{S,out}$, and the corresponding spaces for the retrieved channels as $\hil{i}^\mathrm{R,in}$ and $\hil{i}^\mathrm{R,out}$. Regardless of the specific causal ordering of the combs, we have
\begin{equation}
    L, ~ L^\mathrm{det} \in \lin \left( \bigotimes_{i=1}^I \hil{i}^\mathrm{S,in}\otimes \hil{i}^\mathrm{S,out} \otimes \hil{i}^\mathrm{R,in} \otimes \hil{i}^\mathrm{R,out} \right).
\end{equation}

Following the same reasoning as in Lemma~\ref{lem:symmetry}, we can assume that the pSAR combs satisfy the symmetry conditions
\begin{equation}
    \label{eq:symmetry_tensor_product}  \left[ L, g_X(V_i) \right] = 0, \quad \left[ L^\mathrm{det}, g_X(V_i) \right] = 0 \qquad (\forall V_i \in U(d_i), ~ i=1,\ldots,I)
\end{equation}
where $X \in \{ \mathrm{in}, \mathrm{out} \}$ and the operator $g_X(V_i)$ is defined as
\begin{equation}
    g_X (V_i) := V_{\hil{i}^\mathrm{S,X}} \otimes \overline{V}_{\hil{i}^\mathrm{R,X}} \otimes \id.
\end{equation}
Here, the identity operator acts on the complement of the subsystem $\hil{i}^\mathrm{S,X} \otimes \hil{i}^\mathrm{R,X}$. Based on the unitary invariance $[E, V_i \otimes \overline{V}_i] = 0$, Schur's lemma implies that any operator $E \in \lin (\hil{i}^\mathrm{S,X} \otimes \hil{i}^\mathrm{R,X})$ is a linear combination of two orthogonal projectors:
\begin{equation}
    P_0^{i,X} := \frac{\dket{\id_{d_i}^\mathrm{X}} \dbra{\id_{d_i}^\mathrm{X}}}{d_i}, \qquad P_1^{i,X} := \id_{\hil{i}^\mathrm{S,X}} \otimes \id_{\hil{i}^\mathrm{R,X}} - P_0^{i,X},
\end{equation}
where $\dket{\id_{d_i}^\mathrm{X}} := \sum_{j=1}^{d_i} \ket{i}_{\hil{i}^\mathrm{S,X}} \otimes \ket{i}_{\hil{i}^\mathrm{R,X}}$.
Consequently, the pSAR combs admit the following decompositions:
\begin{equation}
    \label{eq:decomposition_L}  L = \sum_{s} c_s \bigotimes_{i=1}^{I} P_{s(i,\mathrm{in})}^{i,\mathrm{in}} \otimes P_{s(i,\mathrm{out})}^{i,\mathrm{out}}, \qquad L^\mathrm{det} = \sum_{s \in \{ 0,1 \}^{I \times X} } c_s^\mathrm{det} \bigotimes_{i=1}^{I} P_{s(i,\mathrm{in})}^{i,\mathrm{in}} \otimes P_{s(i,\mathrm{out})}^{i,\mathrm{out}},
\end{equation}
where the summation is over all $s \in \{ 0,1 \}^{I \times \{ \mathrm{in}, \mathrm{out} \}} $.
The coefficients $c_s$ and $c_s^\mathrm{det}$ must be non-negative,as the combs are positive semidefinite and the basis elements are mutually orthogonal projectors.

Furthermore, following the same reasoning as in Lemma \ref{lem:identity_is_enough}, the necessary and sufficient condition for a comb with symmetry \eqref{eq:symmetry_tensor_product} to realize the pSAR of product unitary channels with probability $p$ (under the given temporal structure) is:
\begin{equation}
    L \star C_{\iden^S} = p ~ C_{\iden^R},
\end{equation}
where $\iden^S$ and $\iden^R$ denote the product of identity channels for the storage and retrieval stages, respectively. The left-hand-side reduces to:
\begin{equation}
    L \star C_{\iden^S} = \sum_{s \in \{ 0,1 \}^{I \times X} } c_s M_s \qquad M_s := \bigotimes_{i=1}^{I} \left( P_{s(i,\mathrm{in})}^{i,\mathrm{in}} \otimes P_{s(i,\mathrm{out})}^{i,\mathrm{out}} \right) \star C_{\iden^S_i},
\end{equation}
where $\iden^S_i$ is the identity channel from $\hil{i}^\mathrm{S,in}$ to $\hil{i}^\mathrm{S,out}$. Since $C_{\iden^R}$ is a rank-1 operator, $c_s$ must vanish whenever $M_s$ is not proportional to $C_{\iden^R}$.
A direct calculation for each $i$ yields:
\begin{align}
    \left( P_{s(i,\mathrm{in})}^{i,\mathrm{in}} \otimes P_{s(i,\mathrm{out})}^{i,\mathrm{out}} \right) \star C_{\iden^S_i} = \left\{
    \begin{array}{ll}
        \frac{1}{d_i^2} C_{\iden^R_i} & s(i,\mathrm{in})=s(i,\mathrm{out})=0, \\
        \frac{1}{d_i} \left( \id_{\hil{i}^\mathrm{R,in}} \otimes \id_{\hil{i}^\mathrm{R,out}} - \frac{1}{d_i} C_{\iden^R_i} \right) & s(i,\mathrm{in}) \neq s(i,\mathrm{out}), \\
        \left( d_i - \frac{2}{d_i} \right) \id_{\hil{i}^\mathrm{R,in}} \otimes \id_{\hil{i}^\mathrm{R,out}} + \frac{1}{d_i^2} C_{\iden^R_i} & s(i,\mathrm{in})=s(i,\mathrm{out})=1.
    \end{array}\right.
\end{align}
Given that $C_{\iden^R} = \otimes_i C_{\iden^R_i}$, it follows that $c_s = 0$ must hold whenever any outcome of $s$ is $1$.
Consequently, we arrive at the simplified expression for $L$:
\begin{equation}
    \label{eq:L_single_term}    L = c \bigotimes_{i=1}^{I} P_0^{i,\mathrm{in}} \otimes P_0^{i,\mathrm{out}}.
\end{equation}
The coefficient $c$ is related to the success probability $p$ via:
\begin{equation}
    \label{eq:prob_and_c}   p = c \prod_{i=1}^I d_i^{-2}.
\end{equation}

We now evaluate the comb conditions \eqref{eq:comb_conditions} and \eqref{eq:probabilistic_comb_condition} to determine the maximum value of the coefficient $c$.
The final two ports of $L^\mathrm{det}$ correspond to $\otimes_{i=1}^I \hil{i}^\mathrm{R,out}$ and $\otimes_{i=1}^I \hil{i}^\mathrm{R,in}$. The deterministic comb condition \eqref{eq:comb_conditions} implies that
\begin{equation}
    \tr_{\otimes_{i=1}^I \hil{i}^\mathrm{R,out}} L^\mathrm{det} = \id_{\otimes_{i=1}^I \hil{i}^\mathrm{R,in}} \otimes L',
\end{equation}
with some deterministic comb $L'$. However, given the decomposition \eqref{eq:decomposition_L} arising from the symmetry, we have
\begin{align}
    \tr_{\otimes_{i=1}^I \hil{i}^\mathrm{R,out}} L^\mathrm{det} &= \sum_{s \in \{ 0,1 \}^{I \times X} } c_s^\mathrm{det} \bigotimes_{i=1}^{I} P_{s(i,\mathrm{in})}^{i,\mathrm{in}} \otimes \tr_{\hil{i}^\mathrm{R,out}} \left[ P_{s(i,\mathrm{out})}^{i,\mathrm{out}} \right] \\
    & \propto \sum_{s \in \{ 0,1 \}^{I \times X} } c_s^\mathrm{det} \bigotimes_{i=1}^{I} P_{s(i,\mathrm{in})}^{i,\mathrm{in}} \otimes \id_{\hil{i}^\mathrm{S,out}} \\
    & = \id_{\otimes_{i=1}^I \hil{i}^\mathrm{S,out}} \otimes \left( \sum_{s \in \{ 0,1 \}^{I \times X} } c_s^\mathrm{det} \bigotimes_{i=1}^{I} P_{s(i,\mathrm{in})}^{i,\mathrm{in}} \right).
\end{align}
For these two expressions for $\tr_{\otimes_{i=1}^I \hil{i}^\mathrm{R,out}} L^\mathrm{det}$ to be consistent, the operator in the parentheses must satisfy
\begin{equation}
    \sum_{s \in \{ 0,1 \}^{I \times X} } c_s^\mathrm{det} \bigotimes_{i=1}^{I} P_{s(i,\mathrm{in})}^{i,\mathrm{in}} = \id_{\otimes_{i=1}^I \hil{i}^\mathrm{R,in}} \otimes L'',
\end{equation}
with some operator $L''$. This is possible only if
\begin{equation}
    \sum_{s \in \{ 0,1 \}^{I \times X} } c_s^\mathrm{det} \bigotimes_{i=1}^{I} P_{s(i,\mathrm{in})}^{i,\mathrm{in}} \propto \id_{\otimes_{i=1}^I \hil{i}^\mathrm{R,in}} \otimes \id_{\otimes_{i=1}^I \hil{i}^\mathrm{S,in}}.
\end{equation}
By considering the normalization of the deterministic comb, we obtain
\begin{align}
    \tr_{\otimes_{i=1}^I \hil{i}^\mathrm{R,out}} L^\mathrm{det} &= \prod_{i=1}^I d_i^{-1} \id_{\otimes_{i=1}^I \hil{i}^\mathrm{S,in}} \otimes \id_{\otimes_{i=1}^I \hil{i}^\mathrm{S,out}} \otimes \id_{\otimes_{i=1}^I \hil{i}^\mathrm{R,in}} \\
    &= \prod_{i=1}^I d_i^{-1} \bigotimes_{i=1}^I \id_{\hil{i}^\mathrm{S,in}} \otimes \id_{\hil{i}^\mathrm{S,out}} \otimes \id_{\hil{i}^\mathrm{R,in}}.
\end{align}
Similarly, taking the partial trace of $L$ in \eqref{eq:L_single_term} yields
\begin{equation}
    \tr_{\otimes_{i=1}^I \hil{i}^\mathrm{R,out}} L = c \prod_{i=1}^I d_i^{-1} \bigotimes_{i=1}^{I} P_0^{i,\mathrm{in}} \otimes \id_{\hil{i}^\mathrm{S,out}}.
\end{equation}
The condition for a probabilistic comb \eqref{eq:probabilistic_comb_condition} requires $\tr_{\otimes_{i=1}^I \hil{i}^\mathrm{R,out}} L \leq \tr_{\otimes_{i=1}^I \hil{i}^\mathrm{R,out}} L^\mathrm{det}$, which leads to
\begin{equation}
    c \bigotimes_{i=1}^{I} P_0^{i,\mathrm{in}} \leq \bigotimes_{i=1}^I \id_{\hil{i}^\mathrm{S,in}} \otimes \id_{\hil{i}^\mathrm{R,in}} = \id_{\otimes_{i=1}^I  \hil{i}^\mathrm{S,in} \otimes \hil{i}^\mathrm{R,in}}.
\end{equation}
Since $\otimes_{i=1}^{I} P_0^{i,\mathrm{in}}$ is a projector, this operator inequality is satisfied if and only if
\begin{equation}
    0 \leq c \leq 1.
\end{equation}

Combining the relation \eqref{eq:prob_and_c} with the constraint $c \leq 1$, we conclude that the maximum success probability is
\begin{equation}
    p = \prod_{i=1}^I d_i^{-2},
\end{equation}
which is attained when $c=1$. This success probability is realized by independently applying probabilistic teleportation for each unitary channel.

\subsection{Numerical comparison with tensor PBT}\label{sec:tensor_product_numerical}

For a unitary staircase that reduces to a tensor-product unitary channel, PBT can be applied independently to each component channel. We refer to this protocol as tensor PBT. 
Its success probability is 
\begin{equation}
    p_N^{\mathrm{tPBT}} := \prod_{k=0}^{K} \frac{N}{N-1+d_{2k}^2}.
\end{equation}
For $N=1$, tensor PBT is optimal by Lemma~\ref{lem:product_rule}.
For $N\geq2$, the optimality of tensor PBT has not been proven, but the numerical results in Tables \ref{tab:tensor-staircase} and \ref{tab:tensor-superchannel} are close to the tensor PBT values.

\begin{table}[htbp]
    \centering
    \begin{tabular}{|c|c|S[table-format=1.6,group-separator={}]|S[table-format=1.6,group-separator={}]|c|}
    \hline
    \( N \) & {staircase type} & $p^{\map{U} \rightarrow \map{U}}_{\mathrm{SDPmax},N}$ & { \( p_N^{\mathrm{tPBT}} \) } & $(p^{\map{U} \rightarrow \map{U}}_{\mathrm{SDPmax},N} - p_N^{\mathrm{tPBT}})/p_N^{\mathrm{tPBT}}$ \\
    \hline
    1 & (2, 2, 2, 2) & 0.06250  & 0.06250  & \( < 10^{-10} \) \\
    1 & (2, 2, 4, 4) & 0.01563  & 0.01563  & \( < 10^{-10} \) \\
    1 & (4, 4, 4, 4) & 0.003906 & 0.003906 & \( < 10^{-10} \) \\
    \hline
    2 & (2, 2, 2, 2) & 0.1600   & 0.1600   & \( 2.6 \times 10^{-5} \) \\
    2 & (2, 2, 4, 4) & 0.04706  & 0.04706  & \( 1.3 \times 10^{-4} \) \\
    2 & (2, 2, 6, 6) & 0.02163  & 0.02162  & \( 3.4 \times 10^{-4} \) \\
    2 & (3, 3, 6, 6) & 0.01082  & 0.01081  & \( 9.7 \times 10^{-4} \) \\
    \hline
    3 & (2, 2, 2, 2) & 0.2500   & 0.2500   & \( 2.8 \times 10^{-4} \) \\
    3 & (2, 2, 4, 4) & 0.08345  & 0.08333  & \( 1.4 \times 10^{-3} \) \\
    \hline
    \end{tabular}
    \caption{The maximum success probabilities of $N$-to-one pSAR of unitary staircases that reduce to tensor-product unitary channels, numerically obtained by solving SDP. The numerically maximized success probabilities $p^{\map{U} \rightarrow \map{U}}_{\mathrm{SDPmax},N}$ are compared to the success probability $p_N^{\mathrm{tPBT}}$ of the tensor PBT protocol.
    }
    \label{tab:tensor-staircase}
\end{table}

\begin{table}[htbp]
    \centering
    \begin{tabular}{|c|c|S[table-format=1.6,group-separator={}]|S[table-format=1.6,group-separator={}]|c|}
    \hline
    \( N \) & {superchannel type} & $p^{\smap{U} \rightarrow \smap{U}}_{\mathrm{SDPmax},N}$ & { \( p_N^{\mathrm{tPBT}} \) } & $(p^{\smap{U} \rightarrow \smap{U}}_{\mathrm{SDPmax},N} - p_N^{\mathrm{tPBT}})/p_N^{\mathrm{tPBT}}$ \\
    \hline
    1 & (2, 2, 2, 2) & 0.06250  & 0.06250  & \( < 10^{-10} \) \\
    1 & (2, 2, 4, 4) & 0.01563  & 0.01563  & \( < 10^{-10} \) \\
    1 & (4, 4, 4, 4) & 0.003906 & 0.003906 & \( < 10^{-10} \) \\
    \hline
    2 & (2, 2, 2, 2) & 0.1600   & 0.1600   & \( 2.3 \times 10^{-5} \) \\
    2 & (2, 2, 4, 4) & 0.04706  & 0.04706  & \( 1.3 \times 10^{-4} \) \\
    2 & (2, 2, 6, 6) & 0.02163  & 0.02162  & \( 4.0 \times 10^{-4} \) \\
    2 & (3, 3, 6, 6) & 0.01082  & 0.01081  & \( 7.0 \times 10^{-4} \) \\
    \hline
    3 & (2, 2, 2, 2) & 0.2500   & 0.2500   & \( 2.4 \times 10^{-4} \) \\
    3 & (2, 2, 4, 4) & 0.08338  & 0.08333  & \( 5.3 \times 10^{-4} \) \\
    \hline
    \end{tabular}
    \caption{The maximum success probabilities of $N$-to-one pSAR of unitary superchannels whose corresponding staircases reduce to tensor-product unitary channels, numerically obtained by solving SDP. The numerically maximized success probabilities $p^{\smap{U} \rightarrow \smap{U}}_{\mathrm{SDPmax},N}$ are compared to the success probability $p_N^{\mathrm{tPBT}}$ of the tensor PBT protocol.
    }
\label{tab:tensor-superchannel}
\end{table}

\section{Reduction of SDP} \label{sec:reduction}
Let $L$ be the probabilistic comb for the pSAR and $L^\mathrm{det}$ be the deterministic comb such that $L \leq L^\mathrm{det}$. From the unitary-equivalence symmetry we have the decompositions \eqref{eq:mixed_SW} for both operators:
\begin{align} \label{eq:4}
    L^\mathrm{det} &= \left( \sum_{\lambda_l \in \hat{\maptp}^{d_l}_{N,1}} \sum_{S_l,T_l \in \mathrm{Paths}(\lambda_l)} \right)_{l=0,\dots,2K+1} 
    c_{(S_l,T_l)_l}^{(\lambda_l)_l} \bigotimes_{l=0}^{2K+1} E_{S_l, T_l}^{\lambda_l},\\
    L &= \left( \sum_{\lambda_l \in \hat{\maptp}^{d_l}_{N,1}} \sum_{S_l,T_l \in \mathrm{Paths}(\lambda_l)} \right)_{l=0,\dots,2K+1} 
    a_{(S_l,T_l)_l}^{(\lambda_l)_l} \bigotimes_{l=0}^{2K+1} E_{S_l, T_l}^{\lambda_l}.
\end{align}
The matrix units $E_{S_l, T_l}^{\lambda_l}$ are in the form $\id_{m_{\lambda_l}(d)} \otimes \ket{S_l} \bra{T_l}$ in the Schur-transformed basis, where $m_{\lambda_l}(d)$ represents the multiplicity. For each combination $(\lambda_1,\ldots,\lambda_{2K+1})$ of irreducible representations, the coefficients define matrices
\begin{align}
    \label{eq:matrixC} C^{(\lambda_0,\ldots,\lambda_{2K+1})} &:= \left( c^{(\lambda_0,\ldots,\lambda_{2K+1})}_{(S_0,\ldots,S_{2K+1}), (T_0,\ldots,T_{2K+1})} \right)_{(S_0,\ldots,S_{2K+1}), (T_0,\ldots,T_{2K+1})},\\
    \label{eq:matrixA}  A^{(\lambda_0,\ldots,\lambda_{2K+1})} &:= \left( a^{(\lambda_0,\ldots,\lambda_{2K+1})}_{(S_0,\ldots,S_{2K+1}), (T_0,\ldots,T_{2K+1})} \right)_{(S_0,\ldots,S_{2K+1}), (T_0,\ldots,T_{2K+1})}.
\end{align}
With these matrices, the matrix representations of $L$ and $L^\mathrm{det}$ in the Schur-transformed basis are given by
\begin{align}
    [L^\mathrm{det}]^\mathrm{Sch} &:= \bigoplus_{\lambda_l \in \hat{\maptp}^{d_l}_{N,1}} \left( \bigotimes_{l=0,\ldots,2K+1} \id_{m_{\lambda_l}(d)} \right) \otimes C^{(\lambda_1,\ldots,\lambda_{2K+1})}, \\
    [L]^\mathrm{Sch} &:= \bigoplus_{\lambda_l \in \hat{\maptp}^{d_l}_{N,1}} \left( \bigotimes_{l=0,\ldots,2K+1} \id_{m_{\lambda_l}(d)} \right) \otimes A^{(\lambda_1,\ldots,\lambda_{2K+1})},
\end{align}
respectively. Consequently, it is evident that the condition $0 \leq L \leq L^\mathrm{det}$ is equivalent to
\begin{equation}
    \label{eq:probabilistic_condition_matrix}   0 \leq A^{(\lambda_0,\ldots,\lambda_{2K+1})} \leq C^{(\lambda_0,\ldots,\lambda_{2K+1})}
\end{equation}
for all combinations $(\lambda_0,\ldots,\lambda_{2K+1})$.

To optimize the pSAR comb $L$ to obtain the maximum success probability, we can focus on the constraints \eqref{eq:comb_conditions}, \eqref{eq:probabilistic_condition_matrix}, and \eqref{eq:pSAR_comb2}, by treating matrices \eqref{eq:matrixC} and \eqref{eq:matrixA} as variables.
Although all constraints can be cast into SDP, the problem size remains substantial, leading to significant computational overhead.
Our numerical analysis is based on a further reduction of the problem, which can be broadly categorized into the following two types:
\begin{itemize}
    \item Representation of constraints \eqref{eq:comb_conditions} in the Schur-transformed basis, and
    \item Elimination of redundant equations in \eqref{eq:pSAR_comb2}.
\end{itemize}

The constraint \eqref{eq:comb_conditions} includes partial trace and tensor-product with identity operators. The following two lemmas are useful to rewrite \eqref{eq:comb_conditions} in the Schur-transformed basis.
\begin{lemm}[\cite{RamWenzl1992378,grinko2023gelfandtsetlinbasispartiallytransposed}] \label{lem:gt_basis_partial_trace}
Let \( \lambda \in \hat{\maptp}^{d}_{k} \), \( S, T \in \mathrm{Paths}(\lambda) \), and let \( E_{S, T}^{\lambda \ (k)} \) denote a matrix unit of \( \maptp^{d}_{k} \). Then, the following holds:
\begin{equation}
    \Tr_{p+q} \left[ E_{S, T}^{\lambda \ (k)} \right] = \frac{m_{\lambda}}{m_{\mu}} E_{\overline{S}, \overline{T}}^{\mu \ (k-1)},
\end{equation}
where \( \overline{S} \) denotes the path obtained by removing the last system from \( S \), i.e., \( S = \overline{S} \circ \lambda \).
\end{lemm}
\begin{lemm} \label{lem:gt_basis_tensor_identity}
Let \( \lambda \in \hat{\maptp}^{d}_{k} \), \( S, T \in \mathrm{Paths}(\lambda) \), and let \( E_{S, T}^{\lambda \ (k)} \) denote a matrix unit of \( \maptp^{d}_{k} \). Then, the following holds:
\begin{equation}
    E_{S, T}^{\lambda \ (k)} \otimes \id_{\csp^d} = \sum_{\mu \in \hat{\maptp}^{d}_{k+1}, \, \lambda \to \mu} E_{S \circ \mu , T \circ \mu}^{\mu \ (k+1)},
\end{equation}
where \( \lambda \to \mu \) indicates that \( \lambda \) is connected to \( \mu \) in the Bratteli diagram.
\end{lemm}
\noindent Lemma \ref{lem:gt_basis_tensor_identity} can be shown using identities of Clebsch-Gordan coefficients.
Using these lemmas, the constraint \eqref{eq:comb_conditions} is written in terms of linear transformations of matrices \eqref{eq:matrixC} and \eqref{eq:matrixA}.

The final constraint to consider is \eqref{eq:pSAR_comb2}, which defines the probability $p$. Using the matrices $A^{(\lambda_0,\ldots,\lambda_{2K+1})}$ this constraint is given by
\begin{equation}
    \label{eq:pSAR_condition_matrix}    \left( \sum_{\lambda_l \in \hat{\maptp}^{d_l}_{N,1}} \right)_{l=0,\dots,2K+1} 
    A^{(\lambda_0,\ldots,\lambda_{2K+1})} \bullet F^{(\lambda_0,\ldots,\lambda_{2K+1})} = p~ C_{\widetilde{\iden}},
\end{equation}
where $\bullet$ represents the summation of element-wise products $A \bullet F = \sum_{ij} [A]_{ij}[F]_{ij}$, and $F^{(\lambda_0,\ldots,\lambda_{2K+1})}$ is an operator-valued matrix defined by
\begin{equation}
    F^{(\lambda_0,\ldots,\lambda_{2K+1})}_{(S_0,\ldots,S_{2K+1}), (T_0,\ldots,T_{2K+1})} := \left( \bigotimes_{l=0}^{2K+1} E_{S_l, T_l}^{\lambda_l} \right) \star C_{\widetilde{\iden}}^{\otimes N}.
\end{equation}

The constraint \eqref{eq:pSAR_condition_matrix} is a $D \times D$ dimensional matrix equation at first glance. However, it is reduced to fewer equations by exploiting the symmetry. For brevity, we consider $1$-slot superchannels ($K=1$). For this purpose, we decompose each port of the identity superchannel into subsystems $I_i$ and $O_i$ for $i=1,2,3$, as shown in Figure~\ref{fig:identity_superchannel} (a).
\begin{figure}
    \centering
    \includegraphics[width=0.7\linewidth]{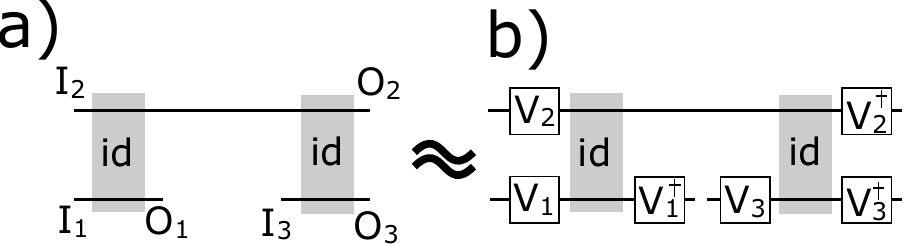}
    \caption{(a) The ($1$-slot) identity superchannel, and (b) a collective unitary transformation that keeps the identity superchannel invariant.}
    \label{fig:identity_superchannel}
\end{figure}
Specifically, applying unitary $V_i$ on $I_i$ and $V_i^\dagger$ on $O_i$ simultaneously but for each $i$ leaves the identity superchannel invariant, as shown in Figure~\ref{fig:identity_superchannel} (b). We represent this invariance in terms of combs as
\begin{equation}
    C_{\map{V}_{O_i}^\dagger \circ \widetilde{\iden} \circ \map{V}_{I_i}} = C_{\widetilde{\iden}}.
\end{equation}
We have
\begin{align}
    F^{(\lambda_0,\ldots,\lambda_{2K+1})}_{(S_0,\ldots,S_{2K+1}), (T_0,\ldots,T_{2K+1})} &:= \left( \bigotimes_{l=0}^{2K+1} E_{S_l, T_l}^{\lambda_l} \right) \star C_{\widetilde{\iden}}^{\otimes N} = \left( \bigotimes_{l=0}^{2K+1} E_{S_l, T_l}^{\lambda_l} \right) \star C_{\map{V}_{O_i}^\dagger \circ \widetilde{\iden} \circ \map{V}_{I_i}}^{\otimes N} \\
    &= (V_{I_i} \otimes \bar{V}_{O_i})^{\otimes N} \left( \bigotimes_{l=0}^{2K+1} E_{S_l, T_l}^{\lambda_l} \right) (V_{I_i} \otimes \bar{V}_{O_i})^{\dagger \otimes N} \star C_{\widetilde{\iden}}^{\otimes N} \\
    &= (V_{I_i} \otimes \bar{V}_{O_i})^\top_\mathrm{R} \left( \bigotimes_{l=0}^{2K+1} E_{S_l, T_l}^{\lambda_l} \right) \overline{(V_{I_i} \otimes \bar{V}_{O_i})}_\mathrm{R} \star C_{\widetilde{\iden}}^{\otimes N} \\
    &= (V_{I_i} \otimes \bar{V}_{O_i})^\top_\mathrm{R} F^{(\lambda_0,\ldots,\lambda_{2K+1})}_{(S_0,\ldots,S_{2K+1}), (T_0,\ldots,T_{2K+1})} \overline{(V_{I_i} \otimes \bar{V}_{O_i})}_\mathrm{R},
\end{align}
where the third line follows from the fact that the matrix units commute with $V^{\otimes N} \otimes \bar{V}$ by definition. We arrive at the symmetry
\begin{equation}
    [F^{(\lambda_0,\ldots,\lambda_{2K+1})}_{(S_0,\ldots,S_{2K+1}), (T_0,\ldots,T_{2K+1})}, V_{I_i} \otimes \bar{V}_{O_i} ] =0, \qquad (\forall \mathrm{unitary}~V, ~i=1,2,3),
\end{equation}
which holds for all elements of $F^{(\lambda_0,\ldots,\lambda_{2K+1})}$. As a result, all the elements of $F^{(\lambda_0,\ldots,\lambda_{2K+1})}$ can be represented as a linear combination of 3-tensored matrix units from \( \maptp^{d_i}_{1,1} \). Since there are only two independent matrix units for \( \maptp^{d_i}_{1,1} \), the number of independent equations in the constraint \eqref{eq:pSAR_condition_matrix} reduces to $2^3=8$.

\end{document}